\documentclass[superscriptaddress,showpacs,twocolumn,nofootinbib,aps,pre]{revtex4-2}
\usepackage{fontawesome5}
\usepackage{comment}
\usepackage{float}
\usepackage[percent]{overpic}
\usepackage[normalem]{ulem}
\usepackage{graphicx}
\usepackage{color}
\usepackage{xcolor}
\usepackage{verbatim}
\usepackage{subfigure}
\usepackage{amsmath}
\usepackage{amssymb}
\usepackage{amsthm}
\usepackage{mathtools}
\usepackage{comment}
\usepackage{chngcntr}
\usepackage{appendix}
\usepackage{lipsum}
\usepackage{dsfont}
\usepackage{xfrac}
\usepackage{xr}
\usepackage[makeroom]{cancel}
\usepackage{ulem}

\newtheorem{theorem}{Theorem}

\newtheorem{proposition}[theorem]{Proposition}

\newcommand{\dd}{\text{d}}

\newcommand{\ee}{\text{e}}

\newcommand{\p}{\partial}

\newcommand{\subfigref}[2]{\hyperref[#1]{\ref*{#1}#2}} % Include subfig letter in link.

\usepackage{amsthm}
\usepackage{bm}
\usepackage[hang,flushmargin]{footmisc}
\usepackage[colorlinks = true,linkcolor = blue,urlcolor = blue,citecolor = blue,anchorcolor = blue]{hyperref}
\definecolor{greencross}{RGB}{83, 170, 46}
\begin{document}

\title{Predicting the conditions for observing the Mpemba effect}
\author{Yue Liu}
\affiliation{Center for Gravitational Physics and Quantum Information, Yukawa Institute for Theoretical Physics, Kyoto University, Kitashirakawa Oiwakecho, Sakyo-ku, Kyoto 606-8502, Japan}

\author{Tan Van Vu}
\affiliation{Center for Gravitational Physics and Quantum Information, Yukawa Institute for Theoretical Physics, Kyoto University, Kitashirakawa Oiwakecho, Sakyo-ku, Kyoto 606-8502, Japan}

\author{Raphaël Chétrite}
\affiliation{CNRS Laboratoire Ypatia des Sciences Mathématiques (LYSM), Piazzale Aldo Moro 5, 00185  Rome, Italy}
%\affiliation{Institut de Physique de Nice (INPHYNI), Université Côte d’Azur \&  CNRS (UMR 7010), 17 rue Julien Lauprêtre, 06200, Nice, France}

\author{Frédéric van Wijland}
\affiliation{Laboratoire Mati\`ere et Syst\`emes Complexes (MSC), Université Paris Cité  \& CNRS (UMR 7057), 75013 Paris, France}
\affiliation{Center for Gravitational Physics and Quantum Information, Yukawa Institute for Theoretical Physics, Kyoto University, Kitashirakawa Oiwakecho, Sakyo-ku, Kyoto 606-8502, Japan}

\author{Hisao Hayakawa}
\affiliation{Center for Gravitational Physics and Quantum Information, Yukawa Institute for Theoretical Physics, Kyoto University, Kitashirakawa Oiwakecho, Sakyo-ku, Kyoto 606-8502, Japan}

\begin{abstract}
The Mpemba effect, a counterintuitive phenomenon where a hotter system relaxes faster than a colder one, has been widely observed in various nonequilibrium systems. 
Despite this progress, the fundamental structural features of the energy landscape required for its emergence remain a subject of debate. 
In this study, we investigate the conditions for the Mpemba effect within one-dimensional overdamped Langevin dynamics. 
We classify the potential landscapes based on the presence of single or double wells, their symmetry properties, and the existence of walls.
We establish that the existence of the effect is primarily driven by the presence of boundaries, either hard or soft, rather than the specific internal structure of the potential landscape, such as metastability or the number of minima. 
By employing a spectral decomposition of the Fokker-Planck operator, we analyze the behavior of the first nontrivial eigenmode and demonstrate that its derivative acts as a Dirac delta peak in the low-temperature regime. This helps us elucidate the mechanism underlying the Mpemba effect: it appears as the interplay between this behavior and the initial population dynamics in a non-trivial way induced by the presence of the wall.
Our analysis provides a unified classification across single- and double-well potentials, highlighting the crucial role of boundary conditions and asymmetry. 
Furthermore, we demonstrate that this framework allows for the engineering of potential landscapes capable of producing multistage Mpemba transitions.
\end{abstract}

\maketitle

%%%%%%%%%%
\section{Introduction}

The Mpemba effect, the counterintuitive phenomenon in which a hotter system relaxes faster than a colder one under identical conditions, has attracted sustained attention since its modern experimental report by Mpemba and Osborne~\cite{mpemba1969cool}. 
Despite early skepticism~\cite{burridge2016questioning,katz2017reply}, it is now widely recognized as a genuine nonequilibrium phenomenon that can arise in a broad class of systems, ranging from classical fluids to stochastic and quantum systems~\cite{Bechhoefer2021,Teza25,Ares25_review}.

Recent theoretical progress has clarified that the Mpemba effect is a fundamental property of relaxation dynamics toward equilibrium (or a steady state), and can be understood in terms of the spectral decomposition of the underlying evolution operator~\cite{Lu17,Klich19}. 
In Markovian systems obeying detailed balance, the relaxation is governed by a hierarchy of eigenmodes of the Fokker--Planck operator, and the Mpemba effect emerges when the projection of the initial condition onto the slowest nontrivial mode depends non-monotonically on the initial temperature~\cite{Lu17,Busiello21}. 
This mechanism has been confirmed in a variety of contexts, including colloidal systems~\cite{kumar2020,kumar2021,Kumar22,Malhotra24,Tian2025}, granular fluids with thermostats~\cite{Santos17,Torrente19,Biswas20,Biswas21,Mompo20,Megias22b,Patron23}, 
optical resonators~\cite{Keller18,Santos20,Patron21}, and inertial suspensions~\cite{Takada21a}. 
The effect has also been identified in a variety of classical stochastic models, including Markovian dynamics~\cite{Lu17, Klich19, Busiello21,Lin22,Pal24,Santos24}, overdamped Langevin systems~\cite{Biswas23a,Biswas_thesis,chetrite2021metastable,walker2021anomalous,Deguenther22}, and non-Markovian models~\cite{Yang22,Strachan25} as well as in several other classical nonequilibrium settings~\cite{Greaney11,Baity-Jesi19,Gonzalez2021,Yang20,Yang22,Gonzalez21,Holtzman22,Megias22a,Biswas23,Pemartin23,Antonov26}. 
See also some general arguments for understanding the condition for observing the Mpemba effect~\cite{Ohga2024,Vu2025}. 

Remarkably, a recent trend is to study the quantum Mpemba effect~\cite{Ares25_review,Strachan25,Nava19,Carollo21,Manikandan21,Ivander23, Ares23,Chatterjee_2023, Chatterjee24,Joshi2024,Shapira2024,Rylands2024,Moroder24,Chalas24,Ares24,Yamashika2024,Liu2024,Wang2024,Nava2024,Longhi2024,Zhang2025,Turkeshi2025,Bao2025,Yu2025,Beato2026,Yamashika2026}, which may be relevant to quantum information processing and thermalization control.
These developments highlight that the Mpemba effect is not restricted to specific microscopic mechanisms, but rather reflects a general feature of nonequilibrium relaxation processes.

%%%%%%%%%%%%%%%%%%%%%%
\begin{figure*}[t]
    \centering
    \includegraphics[width=1.0\textwidth]{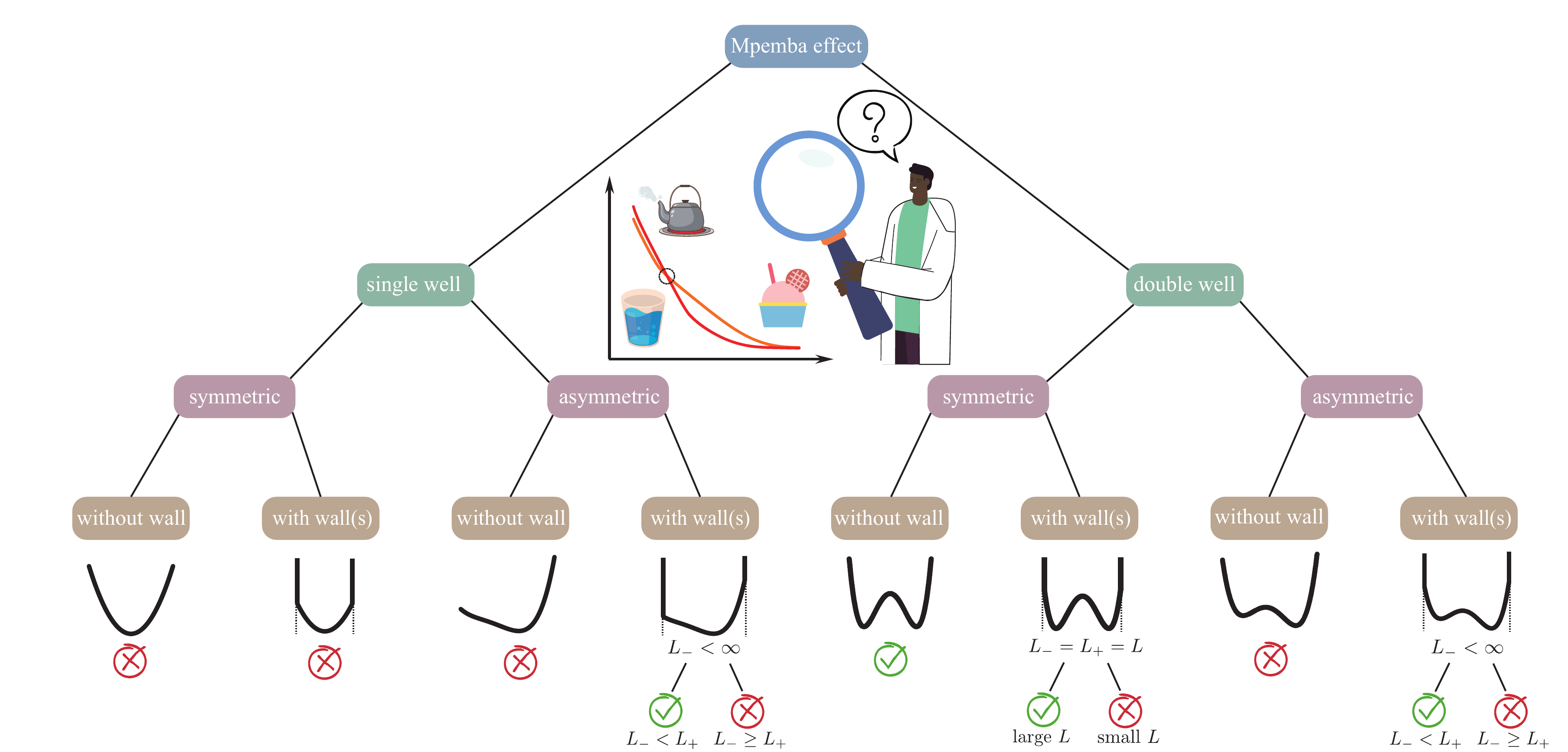}
    \caption{Summary of the conditions for observing the Mpemba effect across various potential landscapes. The decision tree categorizes one-dimensional potentials based on the number of wells (single well vs. double well), their parity (symmetric vs. asymmetric), and the boundary conditions [without wall vs. with wall(s)]. 
    The green checkmark and red cross denote the guaranteed presence and absolute absence of the Mpemba effect, respectively. 
    A simultaneous checkmark and cross mean that the existence of the effect is conditionally dependent on the position of the wall(s). 
    The potential switches to higher-energy tails at $\pm L_{\pm}$. 
    }
    \label{fig:big_picture}
\end{figure*}
%%%%%%%%%%%%%%%%%%%%%%%%%

Despite significant progress, several fundamental questions remain open. In particular, it remains unclear which structural features of the energy landscape are essential to the emergence of the Mpemba effect. 
While early studies emphasized the role of multiple minima and energy barriers, motivated by the original freezing problem, more recent works have suggested that metastability alone is not sufficient and that more subtle aspects of the dynamics must be taken into account~\cite{Klich19,Holtzman22}. 
Moreover, single-well potentials, which do not possess the traditional features of metastability, have been shown to exhibit the Mpemba effect~\cite{Biswas23a,Biswas_thesis}. 
Therefore, a comprehensive understanding of the Mpemba effect that goes beyond the traditional picture of barrier crossing and metastability is still lacking.

% A complementary perspective has been developed in recent works, including our companion paper~\cite{Yue26}, which argues that the presence of boundaries (or ``walls'') plays a decisive role in triggering the Mpemba effect.
A complementary perspective has been developed in recent works, including our companion paper~\cite{Yue26}, which focuses on asymmetric double-well potentials with hard walls and argues that the presence of high-energy potential tails plays a decisive role in triggering the Mpemba effect.
This setting corresponds to the two rightmost cases in Fig.~\ref{fig:big_picture}.
In particular, it was shown that even in simple polynomial potentials, the Mpemba effect can arise due to the interplay between the spatial confinement and the temperature dependence of the probability distribution, rather than the detailed structure of local minima. 
This insight challenges the conventional intuition that associates the effect primarily with double-well or metastable systems.
%See also the two-dimensional counterpart~\cite{Hayakawa2026}.

% In parallel, studies of low-dimensional model systems have provided valuable analytical insights into the mechanism of anomalous relaxation. 
In particular, bistable potentials have served as a minimal setting to explore the interplay between barrier crossing, spectral properties, and initial conditions~\cite{Chatterjee_2023,Chatterjee24}. 
Hayakawa and Takada~\cite{Hayakawa2026} demonstrated that even in two-dimensional bistable systems, the Mpemba effect can exhibit rich behavior depending on the geometry of the potential landscape and the choice of initial ensemble.

In this paper, we build on these developments and investigate the conditions under which the Mpemba effect arises in overdamped Langevin dynamics in one dimension. 
While existing literature examines both single- and double-well potentials, we argue that the specific internal structure of the potential, whether flat, linear, or characterized by nonzero curvature, is not of primary significance.

%The existing literature considers different footing potentials that exhibit one well or two wells. 
%We argue that the primary significance does not lie in the internal structure of the potential, regardless of whether the wells are flat, linear, or characterized by nonzero curvature.

Central to our analysis is the concept of a wall, defined as a rigid boundary external to the potential center. We further extend this framework by introducing soft walls, which arise when the asymptotic growth of the potential at infinity exceeds that predicted by naive analytic continuation. 
We investigate the interplay between these hard and soft walls and the potential's internal structure. 
By analyzing the spectral properties of the Fokker-Planck operator and the behavior of the first nontrivial eigenmode, we elucidate the mechanisms driving the Mpemba effect and derive general conditions for its emergence.

Figure~\ref{fig:big_picture} summarizes our findings.
We provide a comprehensive classification of one-dimensional potentials based on the number of wells (single versus double), parity (symmetric versus asymmetric), and boundary conditions.  
For single-well systems, the Mpemba effect necessitates an asymmetric potential in the presence of at least one wall. Double-well systems exhibit more nuanced behavior: the effect emerges not only in asymmetric potentials with at least one wall but also in symmetric potentials.
Our analysis indicates that the presence of a wall is generally a crucial prerequisite for the Mpemba effect, independent of the potential's internal structure. 
This unified framework reconciles previously disparate observations across both single- and double-well systems. 
Notably, Fig.~\ref{fig:big_picture} highlights potentials where the occurrence of the Mpemba effect is sensitive to the precise location of the wall; this sensitivity is discussed in subsequent sections.

It should be noted that most papers on the quantum Mpemba~\cite{Ares25_review,Strachan25,Nava19,Carollo21,Manikandan21,Ivander23, Ares23,Chatterjee_2023, Chatterjee24,Joshi2024,Shapira2024,Rylands2024,Moroder24,Chalas24,Ares24,Yamashika2024,Liu2024,Wang2024,Nava2024,Longhi2024,Zhang2025,Turkeshi2025,Bao2025,Beato2026,Yamashika2026}, and some classical papers~\cite{Santos17,Torrente19,Biswas20,Biswas21,Mompo20,Megias22b,Patron23,Santos20,Patron21,Takada21a,Biswas_thesis,Deguenther22,Gonzalez2021,Megias22a,Biswas23} use differences in relaxation rates under nonequilibrium initial conditions to observe crossings of time-dependent observables.
However, this paper focuses on the late-stage dynamics, starting from two equilibrium initial conditions, and examines crossings of observables.

The remainder of this paper is organized as follows. In Sec.~\ref{SecII}, we not only introduce the model and the spectral framework for analyzing relaxation dynamics but also recall the mathematical definition of the Mpemba effect. In Sec.~\ref{sec:single}, we investigate single-well potentials and demonstrate the crucial role of walls, both hard and soft. In Sec.~\ref{sec:double}, we extend the analysis to double-well potentials and show how the analysis naturally connects to that of the single-well case. 
We also show how to engineer a potential able to produce multistage Mpemba effects (i.e., multiple Mpemba temperatures). 
Both Secs.~\ref{sec:single} and~\ref{sec:double} are based on numerical results and analytical results with walls located sufficiently far away from the center of the potential. 
%To bypass this analytic limitation, in Sec.~\ref{sec:piecewise}, we analyze a piecewise quadratic potential, which is an exactly solvable double-well potential.
Finally, we summarize our findings and discuss open questions in Sec.~\ref{sec:conclusion}.
We present several appendices to explain the details and assist the main text.

%%%%%%%%%%%%%%%%%%%%%%
\section{Overdamped Langevin dynamics in one dimension}\label{SecII}
%%%%%%%%%%%%%%%%%%%%%%%%

\subsection{The Fokker-Planck evolution operator and its spectral properties}\label{sec:FP}
%%%%%%%%%%%%

In our study, we consider an overdamped Brownian particle moving in a one-dimensional potential $V(x)$ with unit mobility, which is in contact with a thermal bath at temperature $T$. The dynamics of the position $x(t)$ of the particle is described by the Langevin equation
\begin{equation}\label{eq:Langevin}
    \frac{\dd x}{\dd t}=-V'(x(t))+\sqrt{2T}\eta(t),
\end{equation}
where $V'(x):=\dd V(x)/\dd x\,$, and $\eta(t)$ is a Gaussian white noise with zero mean and correlations $\langle \eta(t)\eta(t')\rangle=\delta(t-t')$. The probability density function of the particle's position, denoted as $p(x,t)$, evolves according to the Fokker-Planck equation
\begin{equation}
    \p_t p(x,t)={\mathbb W}p(x,t),
\label{eq:FP}
\end{equation}
where the evolution (Fokker-Planck) operator $\mathbb{W}$ is defined as
\begin{align}\label{op_W}
    \mathbb{W}\cdot :=\p_x(V'(x)\cdot)+T\p_x^2 \cdot .
\end{align}
The operator $\mathbb{W}$ in Eq.~\eqref{op_W} has non-positive eigenvalues $\{-\lambda_n\}_{n\geq 1}$, which satisfy $0= \lambda_1<\lambda_2<\cdots$, and we denote by $\{r_n(x)\}_{n\geq 1}$ the corresponding right eigenvectors of $ \mathbb{W}$. 
Since this problem can be mapped onto the one-dimensional Schr\"{o}dinger equation, there is no degeneracy of the eigenvalues.
Note that the zero eigenmode $r_1(x)$ corresponds to the equilibrium state~\cite{risken1996fokker}.
Given an initial distribution $p(x,0)$, the solution of the Fokker-Planck equation reads
\begin{equation}
    p(x,t)=\sum_{n\ge 1}\ee^{-\lambda_n t}a_n r_n(x),
\end{equation}
where $r_n(x)$ is the $n$th eigenvector, and the coefficients $a_n$ are determined by the initial condition $p(x,0)$
\begin{equation}\label{def:a_n}
    a_n:=\int \dd x\, \ell_n(x)p(x,0),
\end{equation}
where $\ell_n(x)$ is the eigenfunction of the adjoint operator $\mathbb{W}^\dagger$ with the same eigenvalues $\{-\lambda_n\}_{n\geq 1}$.
The value of $\ell_{n}(x)$ represents the contribution of the state $x$ to the $n$th mode of relaxation, and it is crucial for understanding the Mpemba effect.
The dynamics in Eq.~\eqref{eq:Langevin} respects the detailed balance conditions with respect to the Boltzmann distribution $\pi(x,\beta):=\ee^{-\beta V(x)}/Z(\beta)$ with $Z(\beta):=\int \dd x\, \ee^{-\beta V(x)}$ and $\beta=1/T$ (the Boltzmann constant $k_\text{B}$ is set to unity). 
It is well known that Eq.~\eqref{eq:FP} can be mapped onto the Schr\"{o}dinger equation~\cite{risken1996fokker} by means of the Darboux transformation \begin{equation}\label{def:g(x,t)}
g(x,t):=\ee^{\beta V(x)/2}p(x,t) ,
\end{equation}
where the Fokker-Planck equation Eq.~\eqref{eq:FP} can be rewritten as
\begin{equation}\label{Schrodinger}
\partial_t g(x,t)%=[T^2\partial_x^2-U_{\text{eff}}(x)]g(x,t)
=-\mathbb{H} g(x,t) ,
\end{equation}
where $\mathbb{H}:= -T{\dd^2}/{\dd x^2}+U_\mathrm{eff}(x)$, and we have introduced the effective potential
\begin{equation}\label{U(x)}
U_{\text{eff}}(x):=\frac{1}{4T}V'(x)^2-\frac{1}{2}V''(x).
\end{equation}
The form Eq.~\eqref{Schrodinger} allows a direct application of the Sturm-Liouville theory:
\begin{align}\label{phi_r_ell}
\varphi_n(x,\beta)=r_n(x)/\sqrt{\pi(x,\beta)}=\ell_n(x)\sqrt{\pi(x,\beta)} ,
\end{align}
has exactly $n-1$ zeroes, where $\varphi_n$ is the $n$th eigenfunction of $\mathbb{H}$ and it also means that the eigenvectors $r_{m}$ and $\ell_n$ satisfy the orthonormality condition
\begin{equation}\label{orthonormal}
    \int \dd x\, \ell_n(x) r_{m}(x)=\delta_{mn} .
\end{equation}
In most of this work, the particle shall be confined to the region $-L _- \leq x \leq L_+$, where $x=\pm L_\pm$ are identified as the positions of the hard walls. 
%However, since the notion of the wall is central in our work, we would like to take some time to clarify what we mean throughout.
However, for our discussion, we shall introduce two types of walls. 
One is a hard wall located at $x=\pm L_\pm$, where the particle has zero probability of being found outside the wall (there is a zero-flux condition at the wall; inside the wall, the repulsive potential is infinite).  
The other one is a soft wall, in which the particle can exist even outside $[-L_-,L_+]$, where the slope of the potential outside $x<-L_-$ and $x>L_+$ is steeper than the bulk region $-L_-\le x\le L_+$. 
In our numerical simulations for a system confined by soft walls, we use a finite box $[-L_\mathrm{BL},L_\mathrm{BR}]$, where $L_\mathrm{BR}$ ($L_\mathrm{BL}$) is larger than $L_+$ ($L_-$).
For most of this work, the potential $V(x)$ is assumed to take a polynomial form that grows near $x=\pm L_\pm$ as
\begin{align}\label{asymptotic_V(x)}
    V(x)\sim x^m,
\end{align}
where $m$ is an even positive integer. This asymptotic behavior holds for $|x|$ large enough, and it can be probed by the particle provided that the walls at $\pm L_\pm$ are sufficiently far away.

A key finding of Ref.~\cite{Yue26} is that the observation of the Mpemba effect generally requires a boundary located on the shallow side of the double-well potential.
In this paper, we discuss certain exceptions and refinements to this condition. 
Before examining how such a boundary influences the presence of the Mpemba effect, it is necessary to establish a rigorous definition of the effect itself.

\subsection{Definition of the Mpemba effect}

In this paper, we focus on the Mpemba effect after a quench from an equilibrium state $p(x,0)=\pi(x,\beta_i)$ to another equilibrium state with a target inverse temperature $\beta$.
%, though many recent works, especially in the quantum realm, consider nonequilibrium initial states. 
The Mpemba effect considers the dependence on the initial temperature of the relaxation process, and on the counterintuitive fact that, for some potentials, the relaxation toward the target bath temperature can be faster starting from a higher initial temperature than from a lower one (as long as the quench is fast enough~\cite{brey1990}, the Langevin equation Eq.~\eqref{eq:Langevin} remains valid for a time-dependent temperature). 
In the long-time limit, the probability density at a time $t$ after the quench reads
\begin{align}\label{eq:pxt}
    p(x,t)&=\pi(x,\beta)+\ee^{-\lambda_2 t}a_2(\beta_{i},\beta) r_2(x) \notag\\
 &\quad    +\ee^{-\lambda_3 t}a_3(\beta_{i},\beta) r_3(x)%\notag\\&\quad  
 + {O}\left(\ee^{-\lambda_4 t}\right) ,
\end{align}
which is valid for $t\gg\lambda_4^{-1}$. The coefficients $a_n$, introduced in Eq.~\eqref{def:a_n}, are computed using $p(x,0)=\pi(x,\beta_i)$. 
Note that the dependence of $p(x,t)$ on $\beta_i$ and $\beta$, and that of $\lambda_n$ on $\beta$, are not explicitly written in Eq.~\eqref{eq:pxt} and later discussions. 
Unless it vanishes for some reason or unless it is otherwise small, the $a_2$ term is the leading correction to the target equilibrium distribution. 
We can thus quantify the relaxation process by the distance of the instantaneous probability distribution $p(x,t)$ to the target equilibrium state $\pi(x,\beta)$.

A standard measure  of how $p(x,t)$ differs from $\pi(x,\beta)$ is the Kullback-Leibler (KL) divergence~\cite{Lu17,Sagawa}, as a typical monotone measure, defined as 
\begin{equation}
    D_{\text{KL}}(p \| \pi) := \int \dd x\, p(x,t) \ln \left[\frac{p(x,t)}{\pi(x,\beta)}\right]. 
\end{equation}
Alternatively, the $L^1$ distance between $p(x,t)$ and $\pi(x,\beta)$ has also been used~\cite{kumar2020}.
See also a more general argument for monotone measures based on the thermomajorization theory~\cite{Vu2025}. 
Since the eigenvalue $\lambda_{n}$ and eigenfunction $r_n(x)$ are independent of the initial state, the distance to the target state is entirely determined by the overlap coefficient $a_n(\beta_{i},\beta)$. 
For the KL divergence, expanding the logarithm to the second order yields the asymptotic behavior
\begin{align}
    D_{\text{KL}}(p \| \pi) &\simeq \int \dd x\, \frac{[p(x,t) - \pi(x,\beta)]^2}{2\pi(x,\beta)}\notag \\
    &\simeq \ee^{-2\lambda_2 t} \frac{a_2(\beta_{i},\beta)^2}{2} \int \dd x\, \frac{r_2(x)^2}{\pi(x,\beta)}.
\end{align}
%If the $L^{1}$ distance is used, its long-time decay is proportional to $|a_2(\beta_{\text{i}})| \ee^{-\lambda_2 t}$.  
%Therefore, a necessary condition to observe the Mpemba effect is the non-monotonicity of $a_2$ against $\beta_{i}$.

Consider two cooling processes starting from $\beta_{i}^{\text{h}}$ and $\beta_{i}^{\text{c}}$ with $\beta_{i}^{\text{h}}<\beta_{i}^{\text{c}}<\beta$. 
Both intuitively and quantitatively, the colder system is initially closer to the target. Following~\cite{Lu17}, we define the Mpemba effect by the condition $|a_2(\beta_{i}^{\text{h}},\beta)| < |a_{2}(\beta_{i}^{\text{c}},\beta)|$. When this condition is met, the hot state $\beta_{i}^{\text{h}}$ yields a smaller amplitude $a_2$ for the KL divergence (or $L^1$ distance), resulting in a counterintuitively shorter distance to the equilibrium distribution $\pi(x,\beta)$ at long times.
Similarly, for two heating processes, we can also define the inverse Mpemba effect.
Thus, a necessary condition for observing the Mpemba effect (including the inverse one) is that the overlap coefficient $a_2$ must exhibit a non-monotonic behavior as a function of the initial temperature $\beta_{i}$. 
In other words, there must exist an inverse temperature $\beta_{\mathcal{M}}$ such that $(\partial a_2/\partial \beta_{i})|_{\beta_i=\beta_\mathcal{M}}=0$~\cite{Lu17}. Such a temperature will hereafter be referred to as the Mpemba temperature, as it signals the transition to the existence of the Mpemba effect.

Interestingly, it turns out that the derivative $\p a_2/\p \beta_i$ has a thermodynamic interpretation
\begin{align}\label{eq:derivativea2}
    \frac{\p a_2}{\p\beta_i}&=%\text{Cov}_{\beta_i}(\ell_2(x),V(x))
   % \notag\\ &=-
    -\langle\ell_2(x)(V(x)-U_i)\rangle_{\beta_i},
\end{align}
where $U_i:=\langle V(x)\rangle_{\beta_i}$ is the internal energy at inverse temperature $\beta_i$, and where $\langle\ldots\rangle_{\beta_i}$ refers to an equilibrium average  with respect to $\pi(x,\beta_i)$.
Equation~\eqref{eq:derivativea2} is obtained by substituting $\partial_{\beta_i}\pi(x,\beta_i)=-[V(x)-U_i]\pi(x,\beta_i)$ into Eq.~\eqref{def:a_n}.
Note that $\ell_2(x)$ does not depend on $\beta_i$.
The same relation holds for any mode $n$, with $a_2$ and $\ell_2$
replaced with $a_n$ and $\ell_n$, respectively.
As a mathematical curiosity, note that Eq.~\eqref{eq:derivativea2} also expresses that $\ell_2(x)$ and $V(x)$ are uncorrelated with the measure $\pi(x,\beta_{\mathcal{M}})$.
The condition Eq.~\eqref{eq:derivativea2}, which is at the core of our study, demands a fine understanding of $\ell_2(x)$, which we now investigate in some detail.

\subsection{The first excited state}\label{subsec:firstexcited}
Because the Fokker-Planck operator $\mathbb{W}$ is a Markov evolution operator, we know that $\lambda_1=0$ and that $\ell_1(x)=1$. 
In addition, thanks to the mapping in Eq.~\eqref{Schrodinger}, the Sturm-Liouville theory applies and the eigenvectors $\ell_n(x)$ and $r_n(x)$ possess exactly $n-1$ nodes. 
Specifically, $\ell_2(x)$ possesses exactly one zero at some location $x_0$. We now show that $\ell_2(x)$ is monotonic.
By convention, we choose $\ell_2(x)$ positive for $x<x_0$ and negative for $x>x_0$. 
From the definition of $\ell_2(x)$, $\mathbb{W}^\dagger\ell_2(x)=-\lambda_2\ell_2(x)$, we have
\begin{equation}\label{eq:defl2}
T\p_x^2\ell_2(x)-V'(x)\p_x\ell_2(x)=-\lambda_2\ell_2(x) .
\end{equation}
Multiplying both sides by $\ee^{-\beta V(x)}$ and integrating over $x$, we obtain
\begin{equation}
  T \ee^{-\beta V(x)} \p_x\ell_2(x) =-\lambda_2\int_{-L_-}^x\dd y\,\ell_2(y)\ee^{-\beta V(y)} ,
\end{equation}
which holds for a hard wall.
For a soft wall, $L_-$ should be replaced with $\infty$.\footnote{For simplicity, most of the descriptions are based on systems confined by hard walls, except for Sec.~\ref{sec:soft_wall}.}
For $x<x_0$, the right-hand side is negative, it reaches a (negative) minimum at $x_0$, and increases monotonically after $x>x_0$. 
Since for $x=L_+$ for a hard wall (and for $x\to \infty$ for a soft wall), it must vanish due to the orthogonality of $\ell_1(x)$ and $\ell_2(x)$ with the Boltzmann weight, then the integral on the right-hand side must remain negative and $\ell_2(x)$ is decreasing for all $x$. 
Unlike $\ell_1(x)$ and $r_1(x)$, few generic properties of $\ell_2(x)$ and $r_2(x)$ are available. 
When we discuss the single- and double-well potentials, we can characterize $\ell_2(x)$ in more detail.

For instance, one generic feature that we shall establish later is that in the low bath temperature regime, as long as the potential grows faster than quadratically, and irrespective of its internal shape, 
$\ell_2(x)$ has a narrow transient region between two saturated domains, and it reduces to a step function in the $T\to 0$ limit. 
%Although the specifics of the temperature dependence of the width $\delta(T)$ are strongly potential-dependent. 
We defer the formal explanation to Secs.~\ref{sec:single} and~\ref{sec:double}, but it is not hard to understand why $\ell_2(x)$ saturates to some constant values at infinity.
%near the boundaries $x\simeq \pm L_\pm$ in the limit $L_\pm\to \infty$. 
The argument only involves the asymptotics of the potential.

Indeed, if the potential grows as in Eq.~\eqref{asymptotic_V(x)} near the walls ($m$ is even), the dominant contribution in Eq.~\eqref{eq:defl2} near the walls should be the second term on the left-hand side, which satisfies
\begin{equation}
- V'(x)\,\partial_x \ell_2(x) \simeq 0 .
\end{equation}
This naively suggests that $\ell_2(x)$ approaches a constant as $|x|\simeq L_\pm ( \gg 1)$.

However, this argument breaks down for a quadratic potential ($m=2$). 
In that case, we have the scaling behaviors $V'(x) \sim x$ and $\partial_x \ell_2 \sim \ell_2 / x$, implying $V'(x)\partial_x \ell_2(x) \sim \ell_2(x)$. 
Therefore, the statement that $\ell_2(x)$ becomes constant at infinity is no longer valid. Instead, $\ell_2(x)$ exhibits an algebraic tail.
The $m=2$ case will be discussed separately (see Appendix~\ref {app:m2} for details). 

For $m>2$, this reasoning can be made more precise. 
In particular, we estimate the asymptotic rate of convergence by retaining both the terms $V'(x)\partial_x \ell_2$ and $\lambda_2 \ell_2(x)$, while neglecting the diffusion term $T \partial_x^2 \ell_2(x)$, under the assumption that
\begin{equation}
|T \partial_x^2 \ell_2| \ll |V'(x)\partial_x \ell_2|, \, |\lambda_2 \ell_2|
\quad \text{as } x \to \pm L_\pm.
\end{equation}
We can rewrite Eq.~\eqref{eq:defl2} as
\begin{align}
   \ln \frac{\ell_2(+\infty)}{\ell_2(x)} &\simeq \lambda_2\int_x^{\infty}\frac{\dd y}{V'(y)}\simeq\frac{\lambda_2}{m(m-2)} x^{2-m} \label{l_2(L_+)},
\end{align}
for $m>2$ in the limit $L_\pm\to \infty$. 
This leads to
\begin{equation}\label{ell_2(x)}
    \ell_2(x )\simeq C_{+\infty}e^{-O(x^{2-m})}\simeq C_{+\infty}+O(x^{2-m}), 
\end{equation}
for $x\gg 1$ if $L_+ \to \infty$.
Similar asymptotics can be derived on the left side. %, where we have introduced $C_{+\infty}=\ell_2(x\to\infty)$. 
%Our preceding analysis assumes the absence of hard walls. Reintroducing these boundaries remains valid, provided they are placed sufficiently far from the origin to justify the large-$x$ asymptotic analysis. 
If the walls are positioned too close to the origin, the system may fail to reach the constant asymptotics required for the $\ell_2(x)$ mode. 
Consequently, while our analytical derivations assume the limit of distant walls, our numerical simulations explore the full range of $L_\pm$ configurations.

\begin{comment}
In the above reasoning, we have worked in the absence of hard walls. 
If these are restored, the above reasoning remains valid if they are first taken to be far away from the origin, and then the large $x$ analysis can be performed. 
If the walls are not located sufficiently far away, then it is possible that the constant asymptotics for $\ell_2$ is not reached. 
In our analytic arguments, we shall always consider the walls to be sufficiently far away, though our numerics shall explore all regimes of $L_\pm$.
\end{comment}

The region connecting the asymptotic saturated values, which is characterized by the width $w(T)$, will be explored separately for the single- and double-well potentials. However, it is natural to assume that $w(T)$ vanishes as $T\to 0$, which, for a very cold bath, means that in practice, $\p_x\ell_2(x)$ is a downward delta peak somewhere in between the asymptotic branches of the potential. This will guide the qualitative understanding that we present in the next subsection.

%%%%%%%%
\subsection{Tracking the Mpemba effect: a qualitative argument}
%%%%%%%%%%

A necessary condition for observing the Mpemba effect is the existence of an initial temperature $T_i=T_{\mathcal M}$ such that $a_2$ reaches an extremum at that temperature. To develop our qualitative argument, we find it convenient to work in the low-bath-temperature regime, for which we (temporarily) accept that $-\p_x\ell_2(x)\simeq\delta(x-x^*)$.
For single-well potentials, $x^*$ is located at the minimum of the potential. 
In contrast, for double-well potentials, $x^*$ is the saddle point between the two wells (we could choose to use $x^*$ as the origin of coordinates, but it could mislead the reader into thinking $x^*$ is a special symmetry point, which in general it is not).
Again, we shall assume that if there is a hard wall at $-L_-$, the probability $\pi(x,\beta_i)$ vanishes for $x<-L_-$. 
We introduce the cumulative distribution function $\Pi(x,\beta_i)=\int_{-L_-}^{x}\dd y\, \pi(y,\beta_i)$, and we perform an integration by parts in Eq.~\eqref{def:a_n} to obtain
\begin{align}
    a_2&= \int_{-L_-}^{L_+}\dd x\,\ell_2(x)\pi(x,\beta_i) \notag\\
    &=\left[\ell_2(x)\Pi(x,\beta_i)\right]_{-L_-}^{L_+}+\int_{-L_-}^{L_+}\dd x\,\Pi(x,\beta_i)(-\p_x\ell_2) \notag\\
    &\to C_{+\infty}+\Pi(x^*,\beta_i) \qquad (\mathrm{as}\quad  L_\pm\to \infty).
\end{align}
Here, we have used Eq.~\eqref{ell_2(x)}; that is, we assume that if there is a hard wall to the right, it is sufficiently far away so that $\ell_2(x)$ has reached a constant value to the right. 
Since $C_{+\infty}$ does not depend on $\beta_i$, the properties of $a_2$ as a function of $T_i$ are encapsulated in those of the cumulative probability to the left of $x^*$, namely $\Pi(x^*,\beta_i)$, as a function of $T_i$. 

When we consider the initial low temperature as $T_i \to 0$, the initial equilibrium distribution $\pi(x, \beta_i)$ effectively reduces to a Dirac delta function centered at the minimum of the potential. 
As $T_i$ increases, the distribution broadens and becomes biased toward the side where the potential grows more slowly.
If the local asymmetry of $V(x)$ around $x^*$ favors the region $x < x^*$, the cumulative probability $\Pi(x^*,\beta_i)$ increases, and thus $a_2$ increases.
If we further increase the initial temperature $T_i$, $\pi(x, \beta_i)$ samples broader regions of the potential landscape. 
If the potential remains softer for $x < x^*$ at higher energies, $a_2$ continues to rise. 
Conversely, if the potential growth becomes steeper on the left than it is on the right at higher energy scales, and the population shifts toward the right of $x^*$. 
This causes a non-monotonic response in $a_2(\beta_i,\beta)$, which constitutes the fundamental mechanism for the emergence of the Mpemba effect.
The wall, either hard or soft, is precisely the ingredient that can induce such a reversal at a higher initial temperature. It modifies the effective growth of the potential on one side of $x^*$, thereby redirecting the population and turning the temperature dependence of $a_2$ from monotonic to non-monotonic.

\begin{comment}
If we choose a low $T_i$, then the probability is essentially a delta peak around the minimum of the potential. 
Increasing the initial temperature leads $\pi(x,\beta_i)$ to spread out, preferentially towards whichever side has the softer potential growth in the vicinity of $x^*$. 
If this softer side is to the left of $x^*$, then $a_2$ increases. 
Upon further increasing the temperature, the probability $\pi$ will start adjusting, beyond the local minimum, to the branches of the potential. 
If the left of $x^*$ remains softer, then $a_2$ keeps increasing, but if, for some reason, the potential to the left becomes steeper than to the right, the Boltzmann population will leak to the right of $x^*$ and $\Pi(x^*, \beta_i)$, and thus $a_2$ will then decrease. That is the root of the existence of the Mpemba effect.
\end{comment}

As is obvious from our qualitative argument, a key ingredient that we rely on is the property that at low temperature, $-\p_x\ell_2(x)$ is essentially a delta peak, which we expect to smoothen and broaden as the bath temperature itself is increased. 
In the regime in which the bath temperature is not low, although we lose our analytic grip, the same mechanism can still be at play.
The emergence of the Mpemba effect is still governed by the competition between the temperature-dependent population bias and the boundary-induced reversal of this bias induced by the wall.
%This is when not only numerical simulations in Secs.~\ref{sec:single} and~\ref{sec:double}, but also exact results are most welcome to address the Mpemba effect beyond the cold bath limit.

%%%%%%%%%%
\subsection{Symmetric potential}\label{subsec:generalsymmetric}
%%%%%%%%%%%

When the potential is symmetric around some abscissa $x^*$, we choose the origin of coordinates at $x^*$ in this section only, thus effectively setting  $x^{*}=0$. The equation for $\ell_2(x)$ also admits $\ell_2(-x)$ as a solution, hence $\ell_2(x)=\pm\ell_2(-x)$ (due to normalization) and the minus sign is selected because $\ell_2$ is monotonic. As a result, $a_{2}=0$ for any symmetric potential, because $\pi(x,\beta_i)$ is a symmetric function. 
Therefore, the relaxation is controlled by the next nonzero $a_n$ coefficient in Eq.~\eqref{eq:pxt}, namely $a_{3}$. 
As we will now show, the coefficient $a_3$ for such a symmetric potential can be expressed in terms of the coefficient $a_2$ of a companion potential obtained by inserting a hard wall at the symmetry abscissa. We now explain that connection.

%%%%%%%%%%%%%%%%%%%%%%%%%%%
\begin{figure}[t]
\begin{center}
\includegraphics[width=0.45\textwidth]{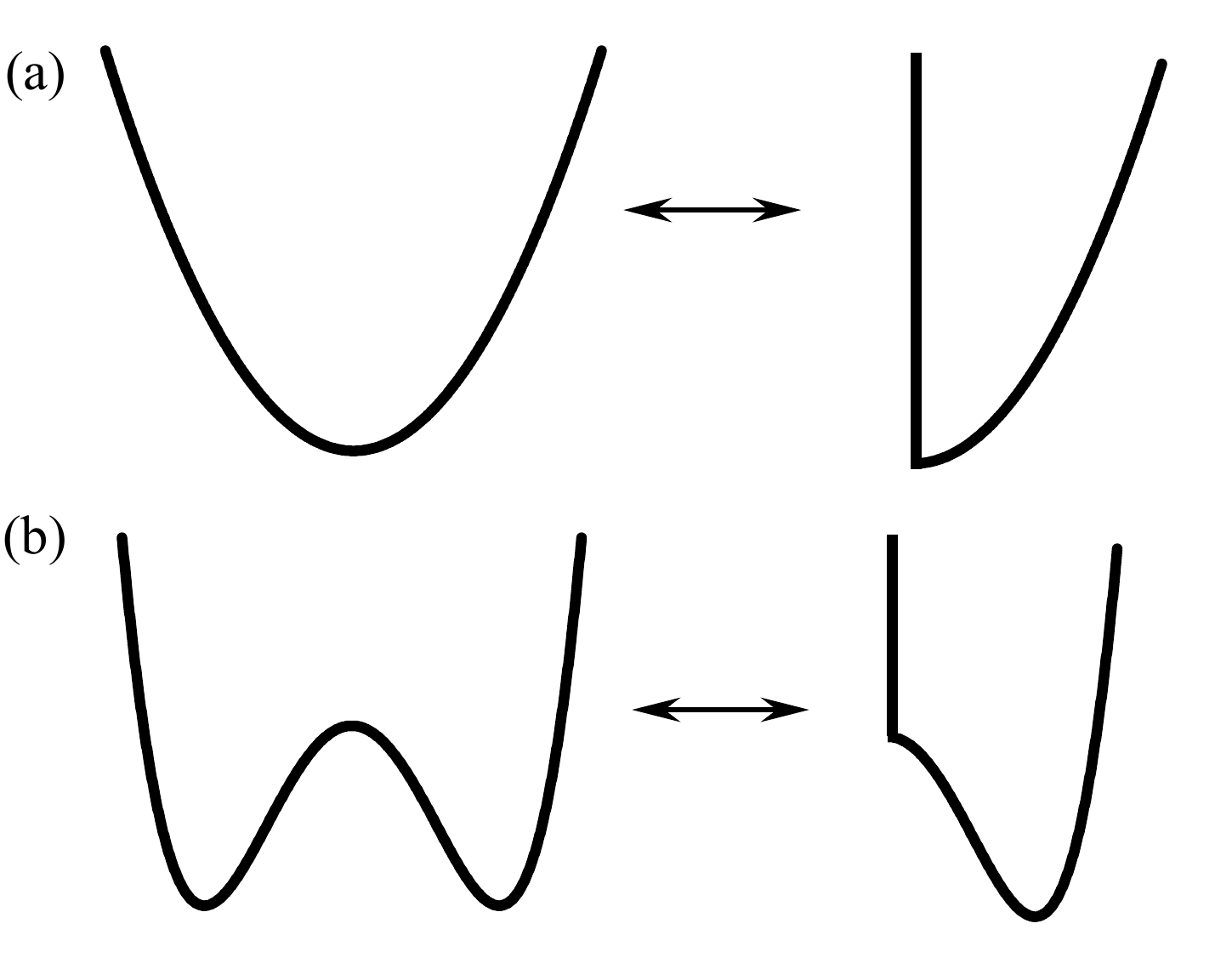}
\caption{(a) A symmetric single-well potential is equivalent to the half-space problem of the same potential with a hard wall placed at the origin ($x=0$). (b) The symmetric double-well potential is equivalent to an asymmetric single-well potential with a wall at the center of the symmetric potential.}
\label{fig:sym_potential}
\end{center}
\end{figure}

We consider an even potential, which is smooth enough (at least \(C^{1}\)) around \(x=0\). 
We now introduce a hard wall at the center of the symmetric potential and another one at \(x=L_+>0\), as illustrated in Fig.~\ref{fig:sym_potential}. 
The new potential coincides with the original one for \(0\le x\le L_+\). 
The zero-flux condition at \(x=0\) leads to the Neumann boundary condition \(\ell_n'(0^+)=0\).

Since the parity operator commutes with the adjoint of the Fokker--Planck evolution operator, each eigenfunction \(\ell_n(x)\) may be chosen to have definite parity, namely, it is either even or odd. 
The parity is, then, fixed by the number of zeros. 
Indeed, \(\ell_n(x)\) has exactly \(n-1\) zeros. 
For an even eigenfunction, all zeros occur in symmetric pairs \(\pm x_j\), so the total number of zeros must be even. By contrast, an odd eigenfunction necessarily vanishes at \(x=0\), and all remaining zeros occur in symmetric pairs, so the total number of zeros must be odd. 
Therefore, \(\ell_n(x)\) is even for odd \(n\) and odd for even \(n\).

Consequently, the even-parity eigenfunctions, namely those with odd index, automatically satisfy the Neumann condition at \(x=0\), because the derivative of an even function vanishes at the origin. Hence, they can be used to construct the eigenfunctions corresponding to the companion potential with a hard wall at \(x=0\). 
For the original symmetric potential, \(\ell_{2n-1}(x)\) has exactly \(2n-2\) zeros for \(n\ge1\). 
Since \(\ell_{2n-1}(x)\) is even, these zeros form \(n-1\) symmetric pairs \(\pm x_j\), so exactly \(n-1\) of them lie on the positive half-axis \(x>0\). 
Therefore, the restriction of \(\ell_{2n-1}(x)\) to \(x\ge0\) satisfies the adjoint equation for the half-space problem, obeys the Neumann boundary condition at \(x=0\), and has exactly \(n-1\) zeros on \(x>0\). 
By rewriting the adjoint eigenvalue equation in Sturm–Liouville form, one can apply the Sturm oscillation theorem, which states that the $n$th eigenfunction has exactly $n-1$ zeros in the open interval. We use this property for both the full symmetric problem and the companion half-space problem.
By this theorem, it is thus the \(n\)-the left eigenvector \(\ell_n^{\mathrm{half}}(x)\) of \(\mathbb{W}\) in the single-well potential embedded in the half-space \(x\ge0\). 
Hence we may choose
\begin{equation}
\ell_n^{\mathrm{half}}(x)=\ell_{2n-1}(x),\qquad x\ge0.
\end{equation}

The normalization condition is also fulfilled, since
\begin{equation}
\begin{split}
&
\int_0^{L_+}\dd x \left[\ell_n^{\mathrm{half}}(x)\right]^2 \pi^{\mathrm{half}}(x,\beta) \notag\\
&=
2\int_0^{L_+}\dd x \left[\ell_{2n-1}(x)\right]^2 \pi(x,\beta)\notag\\
&=
\int_{-L_+}^{L_+}\dd x \left[\ell_{2n-1}(x)\right]^2 \pi(x,\beta)
=1,
\end{split}
\end{equation}
where we have used \(\pi^{\mathrm{half}}(x,\beta)=2\pi(x,\beta)\) for \(x>0\). 
Therefore, the overlap coefficient for the original potential satisfies
\begin{equation}
a_{2n-1}=a_n^{\mathrm{half}},
\end{equation}
which shows that the overlap coefficient for the half-space potential is identical to that of the full symmetric potential. 
The study of the Mpemba effect in a symmetric potential, and in particular of the coefficient \(a_3\), is therefore reduced to the study of the more conventional Mpemba effect associated with \(a_2^{\mathrm{half}}\) in the companion potential obtained by inserting a hard wall at the origin.

%%%%%%%%%%%%%%%%
\section{Single well}\label{sec:single}
%%%%%%%%%%%%%%%%
\subsection{Width of $\ell_2(x)$}\label{subsec:l2single}
%%%%%%%%%%%%%%%%%

For a single-well potential $V (x)$ with a minimum at $x = x^*$, the behavior of $\ell_2(x)$ around $x=x^*$ can be explored in greater detail, and it depends on how $V(x)$ behaves at $x^*$. 
Note that we adopt the general argument, which can cover the case $x^*\ne 0$. 
We assume that for $x\to x^*$, $V (x)\sim x^{m'}$, with an even integer $m'$ ($m'=2$ is the standard case of a locally harmonic minimum). 
This $m'$, of course, can be different from $m$ in Eq.~\eqref{asymptotic_V(x)}.
We have already explained that $\ell_2(x)$ is saturated for $x\to \pm L_\pm$ with sufficiently large $L_\pm$, if $V(x)$ grows faster than quadratically at infinity. 
When the bath temperature is low, we shall show that, in addition, $\ell_2(x)$ takes the form of a step function
whose width $w(T)$ approaches zero in the limit $T\to 0$.
To see this, using rescaled variables, we rewrite Eq.~\eqref{eq:defl2} as
\begin{equation}
\hat{\ell}_2(z) = \ell_2(w z),\,\,\hat{\lambda}_2=w^2T^{-1}\lambda_2,
\end{equation}
which leads to
\begin{equation}\label{scaled_ell_2}
Tw^{-2}\p_z^2\hat{\ell}_2-V'(zw)w^{-1
}\p_z\hat{\ell}_2=-Tw^{-2}\hat{\lambda}_2\hat{\ell}_2 .
\end{equation}
If we assume $w(T)$ does not approach zero in the $T\to 0$ limit, only the second term on the left side survives.
This leads to a constant $\ell_2(x)$, which contradicts the existence of one node from the Sturm-Liouville theory; hence, we must have that $w(T)\to 0$ as $T\to 0$. 
For $w\to 0$ in Eq.~\eqref{scaled_ell_2}, $w\sim T^{1/m'}$ is the only possible choice to remove the singularities in the limit $T\to 0$. 
In this case, Eq.~\eqref{scaled_ell_2} reduces to
\begin{equation}\label{reduced_ell_2}
\p_z^2\hat{\ell}_2-{m'} z^{m'-1}\p_z\hat{\ell}_2=-\hat{\lambda}_2\hat{\ell}_2 .
\end{equation}
The solution of Eq.~\eqref{reduced_ell_2} is given by an orthogonal polynomial, such as the first Hermite polynomial for $m'=2$ (a linear function around $x=0$), and some other polynomial for $m'\geq 4$ obtained by an appropriate use of the Rodrigues formula~\cite{TynMyint-U-1978}. 
We have now established that $\ell_2(x\to\pm\infty)=C_{\pm\infty}$ and that the transition region connecting the two asymptotic plateaus has a width $w(T)$ that vanishes as $T\to0$. 
These results directly lead to $-\p_x\ell_2(x)\simeq\delta(x-x^*)$ in our setting in the limit $T\to 0$. We now exploit the consequences of this result for the analysis of the behavior of $a_2$.

\subsection{Single-well potential between two hard walls}\label{subsec:derivationsinglewall}

For the discussion of this subsection, we introduce two hard walls at $x=-L_-$ and $x=L_+$, which confine the particle to the region $-L_-<x<L_+$. 
Following the argument in Ref.~\cite{Yue26}, where we discuss the double-well potential, we now show how $\beta_{\mathcal{M}}$, defined as $\partial a_2/\partial \beta_i|_{\beta_i=\beta_{\mathcal{M}}}=0$, depends on the positions of the walls, and how the Mpemba effect disappears when the walls are pushed to infinity.

According to Eq.~\eqref{eq:derivativea2}, $\beta_{\mathcal{M}}$ is determined by
\begin{equation}\label{eq:definitionT_M}
    \langle \ell_{2}(x)V(x)\rangle_{\beta_\mathcal{M}}-\langle \ell_{2}(x)\rangle_{\beta_\mathcal{M}}\langle V(x)\rangle_{\beta_\mathcal{M}}=0.
\end{equation}
When the bath temperature is low, the corresponding $\ell_2$ is a step function $\ell_2(x)\simeq\theta(x^*-x)$, hence we have
\begin{equation}\label{U_-p_-}
    U_{-}p_{-}-p_{-}(U_{-}p_{-}+U_{+}p_{+})=0,
\end{equation}
where we have introduced
\begin{equation}\label{eq:pandE}\begin{split}
    U_{\pm}:=&\pm\frac{1}{p_{\pm}}\int_{x^*}^{\pm L_\pm}\dd x\, V(x)\pi(x,\beta_{i}),\\
    p_{\pm}:=&\pm\int_{x^*}^{\pm L_\pm}\dd x\, \pi(x,\beta_{i}).
\end{split}\end{equation}
Consequently, the Mpemba temperature can be found by solving
\begin{equation}\label{U_+=U_-}
    U_{-}=U_{+}.
\end{equation}
We first restrict our analysis to a polynomial potential between the two walls $-L_-<x<L_+$. This excludes potentials that are non-analytic (we focus on potentials that can be written as a series expansion valid everywhere), such as the piecewise harmonic potential, or the potential that exhibits a power law with non-integer $m$. 
This case will be commented upon further down. 
Regardless of the internal structure of the potential (one well, two wells, etc.), by an appropriate choice of the $v_k$'s, such a polynomial can be written in a unique way as
\begin{equation}\label{potential_form}
    V(x)=v_m(x-x^*)^{m}+\sum_{k=0}^{m-1}v_{k}(x-x^*)^{k} 
\end{equation}
for any $x^*$. 
For an asymmetric potential, $m$ must be an even integer larger than two.
For sufficiently large $L_\pm$, the potential remains confining between the walls, scaling as $x^m$ for large $x$.
With such an expression for $V(x)$ as in Eq.~\eqref{potential_form} we see that
\begin{equation}
        (x-x^*)V'(x)=mV(x)-\sum_{k=0}^{m-1}(m-k)v_{k}(x-x^*)^{k}.
\end{equation}

%\textcolor{red}{How can we treat non-integer $n$?}\textcolor{orange}{YL:$n$ should be even. For Odd $n$, the system is unstable. For non-integer $n$, negative halves are not defined for many cases, and the system is also unstable or has singularities.}

Let us go back to Eq.~\eqref{U_-p_-} with Eqs.~\eqref{eq:pandE} and~\eqref{potential_form}.
By using the integration by parts, %\textcolor{red}{As stated later, there is a big gap to introduce $L_\pm$.} 
we can get
\begin{equation}\label{eq:stephard}
    \begin{split}
        Z(\beta_i)p_{-}=&(L_{-}+x^*)\ee^{-\beta_{i} V(-L_-)}+\beta_{i}mU_{-}Z(\beta_i)p_{-}\\&-\beta_{i}\sum_{k=0}^{m-1}(m-k)v_{k}\int_{-L_-}^{x^*}\dd x\,(x-x^*)^{k}\ee^{-\beta_{i} V(x)},\\
        Z(\beta_i)p_{+}=&(L_{+}-x^*)\ee^{-\beta_{i} V(L_+)}+\beta_{i}mU_{+}Z(\beta_i)p_{+}\\&-\beta_{i}\sum_{k=0}^{m-1}(m-k)v_{k}\int_{L_+}^{x^*}\dd x\,(x-x^*)^{k}\ee^{-\beta_{i} V(x)}.
    \end{split}
\end{equation}
%\textcolor{red}{You have suddenly introduced (the location of walls) $L_\pm$. What is its definition? 
%Indeed, Eq.~\eqref{eq:pandE} considers the system from $x=-\infty$ to $\infty$.}
Thus, we obtain
\begin{equation}\label{eq:Uhardwall}
    \begin{split}
        U_{-}=&\frac{1}{m\beta_{i}}+\frac{\sum_{k=0}^{m-1}(m-k)v_{k}\int_{-L_-}^{x^*}\dd x\,(x-x^*)^{k}\ee^{-\beta_{i} V(x)}}{m Z(\beta_i)p_{-}}\\
        &-\frac{(L_{-}+x^*)\ee^{-\beta_i V(-L_-)}}{m\beta_i Z(\beta_i)p_{-}},\\
        U_{+}=&\frac{1}{m\beta_{i}}+\frac{\sum_{k=0}^{m-1}(m-k)v_{k}\int_{x^*}^{L_{+}}\dd x\,(x-x^*)^{k}\ee^{-\beta_{i} V(x)}}{m Z(\beta_i)p_{+}}\\
        &-\frac{(L_{+}-x^*)\ee^{-\beta_i V(L_+)}}{m\beta_i Z(\beta_i)p_{+}}.
    \end{split}
\end{equation}
For large enough $L_\pm$, the condition Eq.~\eqref{U_+=U_-} leads the temperature $\beta_{\mathcal{M}}$ to satisfy
\begin{equation}
    \frac{L_-\ee^{-\beta_{{\mathcal M}} V(-L_-)}}{m\beta_{{\mathcal M}}p_{-}}-\frac{L_+\ee^{-\beta_{{\mathcal M}} V(L_+)}}{m\beta_{{\mathcal M}}p_{+}}
    =\sum_{k=0}^{m-1}\frac{m-k}{m}v_{k}\Delta\langle x^{k}\rangle,
\end{equation}
where $\Delta\langle x^{k}\rangle$ is defined as
\begin{equation}
    \Delta\langle x^{k}\rangle:=\langle x^{k}\rangle_{+}-\langle x^{k}\rangle_{-}
\end{equation}
with
\begin{align}
   \langle x^{k}\rangle_{+}&:=\frac{\int_{x^*}^{L_+}\dd x\,(x-x^*)^{k}\ee^{-\beta_{{\mathcal M}} V(x)}}{p_{+}} ,\\
   \langle x^{k}\rangle_{-}&:=\frac{\int_{-L_-}^{x^*}\dd x\,(x-x^*)^{k}\ee^{-\beta_{\mathcal M} V(x)}}{p_{-}} .
\end{align}    
Our task now is to determine the order of $\Delta\langle x^{k}\rangle$ and $p_{\pm}$ in the large $L_-$ and $L_+$ limits. We anticipate that the regime of interest is the high-temperature one with small $\beta_{\mathcal M}$, yet with $\beta_{\mathcal M}V(-L_-)\gg 1$. In this limit, we have
\begin{equation}
    \begin{split}
    p_-&\simeq M_{0}-\frac{1}{m\beta_{\mathcal M}L_-^{m-1}}\ee^{-\beta_{\mathcal M} L_-^{m}},\\
    p_+&\simeq M_{0}-\frac{1}{m\beta_{\mathcal M}L_+^{m-1}}\ee^{-\beta_{\mathcal M} L_+^{m}},
    \end{split}
\end{equation}
where $M_{k}=\int_{0}^{\infty}\dd x\,x^{k}\ee^{-\beta_{{\mathcal M}} x^{m}}=\beta_\mathcal{M}^{-(k+1)/m}\Gamma((k+1)/m)/m \propto \beta_{\mathcal M}^{-(k+1)/m}$ where $\Gamma(x)$ refers to the Gamma function. 

The second term of each line on the right-hand side of Eq.~\eqref{eq:Uhardwall} is the contribution from the wall, which is negligible compared to $M_0$ in the large $L_\pm$ limit. 
Similar arguments apply to $\langle x^{k}\rangle_{\pm}$, and we have
 \begin{equation}
    \begin{split}
        &\langle x^{k}\rangle_{+}\sim \frac{M_{k}}{M_{0}},\qquad \langle x^{k}\rangle_{-}\sim (-1)^k\frac{M_{k}}{M_{0}}.
    \end{split}
\end{equation}
which suggests that the odd term survives while the even term vanishes for $\Delta\langle x^{k}\rangle$. For definiteness, we assume that the potential is softer to the left of $x^*$ than to the right (of course, if the reverse is true, the roles of $L_\pm$ are swapped into those of $L_\mp$). Hence, when determining the asymptotics
\begin{equation}
\Delta\langle x^{k}\rangle\sim\begin{cases}
    \beta_{i}^{-k/m}&~\text{for odd}~k,\\
    0&~\text{for even}~k .
\end{cases}
\end{equation}
This means the corresponding proportionality coefficient is assumed to be positive. 
This expression shows that only the highest odd term $k^*$ needs to be considered (this is $k^*=m-1$ if $v_{m-1} \neq 0$). 
Therefore, we have
\begin{equation}\label{eq:moment}
    \frac{L_-\ee^{-\beta_{\mathcal M} L_-^{m}}}{m\beta_{\mathcal M}\beta_{\mathcal M}^{-1/m}}-\frac{L_+\ee^{-\beta_{\mathcal M} L_+^{m}}}{m\beta_{\mathcal M}\beta_{\mathcal M}^{-1/m}}\sim \beta_{\mathcal M}^{-k^*/m},
\end{equation}
which can be rewritten as
\begin{equation}\label{eq:transition_temp}
    \beta_{\mathcal M}\simeq\frac{\ln (\text{const}\,L_-)}{L_-^m}+\frac{1}{L_-^m}\ln\left(1-\frac{L_+}{L_-}\ee^{-\beta_{\mathcal M}(L_+^m-L_-^m)}\right),
\end{equation}
irrespective of whether $k^*=m-1$ or a smaller integer.  
Note that Eq.~\eqref{eq:transition_temp} is an implicit equation for $\beta_{\mathcal M}$. 
Further progress can be made when $1\ll L_-\ll  L_+\to \infty$.
Then, we can ignore the second term on the right-hand side of Eq.~\eqref{eq:transition_temp}, and we arrive at
\begin{equation}\label{eq:transition}
    \beta_{\mathcal M}\propto \frac{\ln L_-}{L_-^m}\,.
\end{equation}
This shows explicitly that $\beta_{\mathcal M}$ vanishes as $L_-\to\infty$, and thus the Mpemba effect disappears without a wall. 
Furthermore, $\beta_{\mathcal M}$ exists if $L_{+}$ is large enough with respect to $L_-$. 
However, if $L_+\leq L_-$, the implicit equation Eq.~\eqref{eq:transition_temp} for $\beta_{\mathcal M}$ does not have a solution anymore, which expresses that the Mpemba effect disappears. 
According to Eq.~\eqref{eq:transition_temp}, if the left wall is located infinitely far away, i.e., $L_-\to \infty$, and the right wall is kept at a finite position, i.e., $L_+<\infty$, the temperature $\beta_{\mathcal M}$ also vanishes. 
%Note that the expression of the Mpemba temperature in the large $L_-$ and $L_+$ limits in the asymmetric single-well potential is the same as that in the asymmetric double-well potential, which suggests the inner structure of the potential plays a minor role in the Mpemba effect.
The results presented in this subsection rely on the assumption that $\ell_2(x)$ is essentially a step function.
%The results presented in this subsection are only valid for a single-well potential increasing faster than quadratically at infinity, since our argument relies on the idea that $\ell_2(x)$ is basically a step function. 
%However, even when $\ell_2(x)$ is not a step function, our results are still applicable to many situations.
The mechanism hinges on whether the population can transfer from the soft branch to the steep branch of the potential as the initial temperature increases, and the wall is precisely the factor that can induce such a reversal at higher initial temperatures.
We now examine how to extend our results to the soft-wall cases. 

% , but in the analytic derivation, we immediately take the $L_-\to+\infty$ limit(in the numerics, however, we are bound to work with a finite $L_-$, and the independence of our predictions on a large enough $L_-$ will have to be demonstrated)

%%%%%%%%%%%%%%%%%%%
\subsection{Softening one of the walls}\label{sec:soft_wall}
%%%%%%%%%%%%%%%%%%

We now replace the hard wall with a soft wall at $x=-L_-$. 
In practice, this means that we now assume that
\begin{equation} 
V(x)\sim(x-x^*)^{q}
\end{equation}
with $q>m$ for $x\leq -L_-$. 
The particle is now allowed to probe the $x<-L_-$ region. 
Scanning through the steps of the hard-wall derivation in subsection~\ref{subsec:derivationsinglewall}, only a few adjustments have to be made. For instance, the counterpart of Eq.~\eqref{eq:stephard} is modified into
\begin{equation}\label{eq:soft}
    Z(\beta_i)p_{\text{tail}} = -(L_- + x^*)e^{-\beta_i V(-L_-)} + \beta_i q Z(\beta_i)p_{\text{tail}} U_{\text{tail}},
\end{equation}
where $p_{\text{tail}}$ and $U_{\text{tail}}$ are defined as\footnote{For theoretical analysis, we regard $L_\mathrm{BL}$ as $\infty$.}
\begin{equation}
    \begin{split}
        p_{\text{tail}}&=\int_{-\infty}^{-L_-}\dd x\,\pi(x,\beta_i),\\
        U_{\text{tail}}&=\frac{1}{p_{\text{tail}}}\int_{-\infty}^{-L_-}\dd x\,V(x)\pi(x,\beta_i).
    \end{split}
\end{equation}
To arrive at Eq.~\eqref{eq:soft}, we have used that $(x-x^*)V'(x) \approx q V(x)$ along with an integration by parts. Thus, we have
\begin{equation}\label{eq:Usoftwall}
    \begin{split}
    U_{-}=&\frac{1}{m\beta_{i}}+\frac{\sum_{k=0}^{m-1}(m-k)v_{k}\int_{-L_-}^{x^*}\dd x\,(x-x^*)^{k}\ee^{-\beta_{i} V(x)}}{m Z(\beta_i)p_{-}}\\
    &-\frac{(L_{-}+x^*)\ee^{-\beta_i V(-L_-)}}{m\beta_i Z(\beta_i)p_{-}}\left(1-\frac{m}{q}\right),
    \end{split}
\end{equation}
where we have used $p_{\text{tail}}/p_{-}\ll 1$. 
The only difference between the expression of $U_{-}$ with a hard wall in Eq.~\eqref{eq:Uhardwall} and the steep power law wall with a potential break is the presence of a factor $(1-m/q)$ in the boundary term. 
As a check, taking $q\to\infty$ reduces to the hard wall result.
As a result, $\beta_{\mathcal M}$ with a steep wall exhibits the same behavior as Eq.~\eqref{eq:transition}. 
This establishes that the Mpemba effect disappears in the absence of walls, even if the wall is soft. 
Note that $q=m$ corresponds to a no-wall situation, where the boundary term in Eq.~\eqref{eq:Usoftwall} vanishes.
Here we consider soft walls with different power-law exponents on the two sides.
Our companion paper~\cite{Yue26} also discusses tails with the same power-law exponent, including the linear tails considered in Ref.~\cite{Sane26}.

\subsection{Symmetric single-well potential}
For the symmetric single-well potential, as explained in subsection~\ref{subsec:generalsymmetric}, we can restrict our analysis to a half potential in which a hard wall has been introduced at the symmetry axis, as shown in Fig.~\ref{fig:sym_potential}(a). 
We use the Fortuin-Kasteleyn-Ginibre inequality~\cite{fortuin1971correlation}, which states that the average of the product of two monotonically increasing observables is larger than the product of their averages. 
To address the case of a symmetric single-well potential, we apply this inequality to $-\ell_2^\text{half}(x)$ and to $V(x)$ on the $x>0$ half-axis, which are indeed increasing functions. 
Hence, the condition in Eq.~\eqref{eq:definitionT_M} can never be fulfilled, and there is thus no possibility of a Mpemba effect.

\begin{figure}[t]
\begin{center}
\includegraphics[width=1\linewidth]{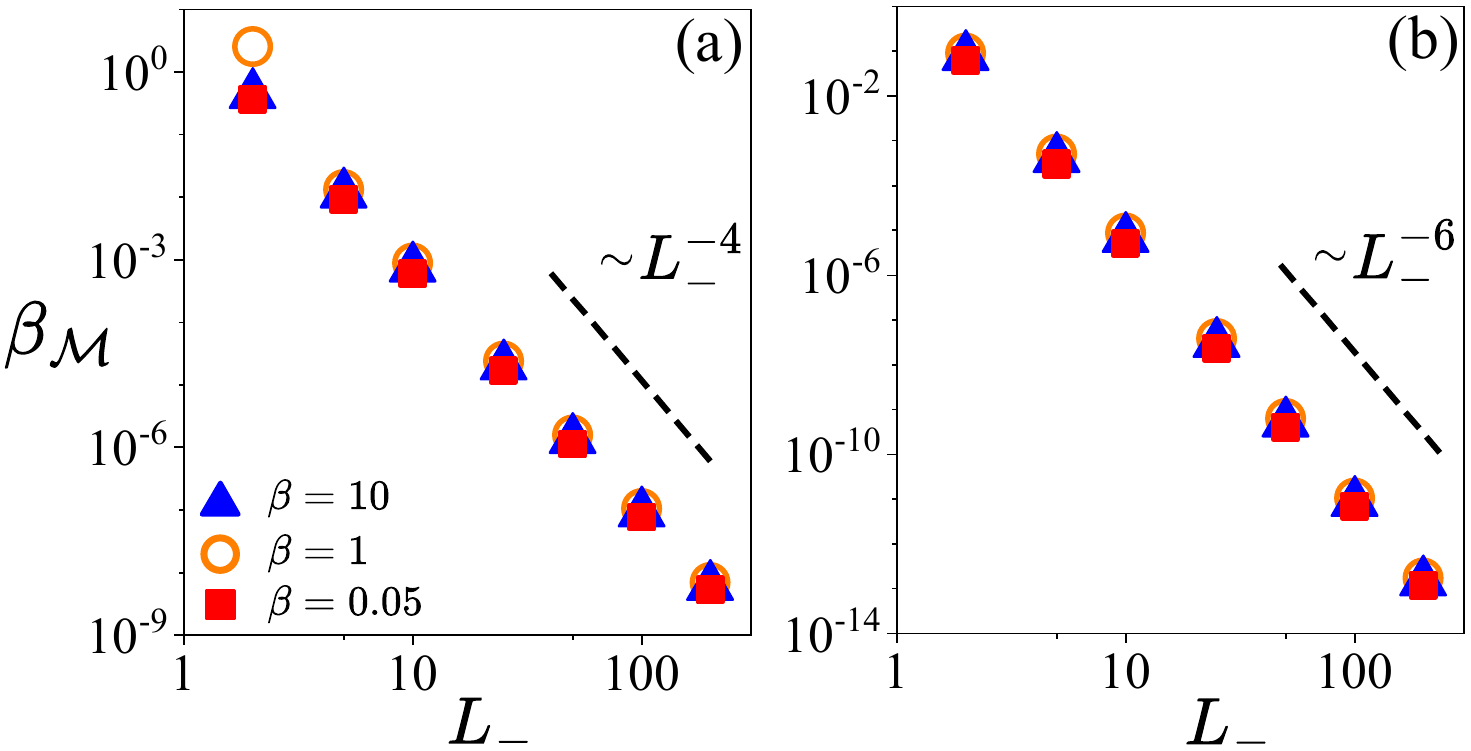}
\caption{ (a)
The inverse temperature $\beta_\mathcal{M}$ is shown as a function of $L_-$ for various bath temperatures in a quartic potential $V(x)\sim x^4$ with a hard wall at $x=-L_-=-L_{\text{BL}}$. (b) Same as (a), but with the quartic potential replaced by a sextic potential $V(x)\sim x^6$. In both (a) and (b), the analytical result, Eq.~\eqref{eq:transition}, obtained for large $\beta$ and $L_-$ is shown as a dashed line as a guide for the eye. Remarkably, no notable deviations are observed, even very far from the regime in which the analytical expression was derived. We have used $L_+=2L_-$ ($L_{\text{BR}}=2L_{\text{BL}}$).}
\label{fig:betaivsbetaandL}
\end{center}
\end{figure}

\subsection{Numerical explorations for a single-well potential}
To explore the Mpemba effect in the single-well potential, we numerically solve the backward Fokker-Planck equation on a finite interval $[-L_{\text{BL}},L_{\text{BR}}]$ to evaluate the overlap coefficient $a_2$ across various bath temperatures. If we are considering two hard walls, $L_{+}$ and $-L_{-}$ are located at $-L_{\text{BL}}$ and $L_{\text{BR}}$. 
If we are considering soft walls, $L_{\text{BL}}$ and $L_{\text{BR}}$ are chosen sufficiently far away, and the independence of our results on the choice of the simulation box is tested. The numerical calculations are performed using the finite-difference method. 
By discretizing the spatial domain $[-L_{\text{BL}},L_{\text{BR}}]$ into an $N = 10^4$ grid points, we obtain an $N \times N$ matrix representation of the evolution operator $\mathbb{W}$. The left eigenvector $\ell_2$ is then extracted by diagonalizing this matrix using the {\tt scipy.linalg} package in Python. Subsequently, the overlap coefficient $a_2$ is calculated via numerical integration of the product of $\ell_2(x)$ and the initial distribution $\pi(x, \beta_i)$. The Mpemba temperature $\beta_{\mathcal M}$ is finally determined by locating the extremum of $a_2$ with respect to $\beta_i$. As an illustrative example, we employ the quartic potential $V(x) = x^4 + x^3 + 0.3x^2$ and the sextic potential $V(x)=x^{6}+x^{5}+0.3x^{4}+0.3x^{2}$, where $x^*=0$ for both potentials.
The right boundary is set to a sufficiently large value $L_+\gg L_-$, which, according to Eqs.~\eqref{eq:transition_temp} and \eqref{eq:transition}, does not affect our numerical results.

Figure~\ref{fig:betaivsbetaandL}(a) plots $\beta_{\mathcal M}$ as a function of the left hard-wall position $-L_-$ for various bath temperatures, ranging from the low- to high-temperature regime. 
As $L_-$ increases, $\beta_{\mathcal M}$ decreases and asymptotically vanishes for sufficiently large $L_-$. 
We also consider a soft-wall scenario, where the potential is replaced with a steeper confinement $V(x) \sim x^6$ for $-L_{\text{BL}}\leq x\leq -L_-$. As depicted in Fig.~\ref{fig:soft_wall}, $\beta_{\mathcal M}$ exhibits the same scaling behavior as the hard-wall case, corroborating our arguments in Sec.~\ref{sec:soft_wall}. 
The results are shown to be independent of the size of the simulation box $L_{\text{BR}}$. 
Furthermore, at high bath temperatures and for a relatively close wall, the Mpemba temperature lies below the bath temperature, a phenomenon known as the inverse Mpemba effect~\cite{Lu17}.

\begin{figure}[t]
\begin{center}
\includegraphics[width=1\linewidth]{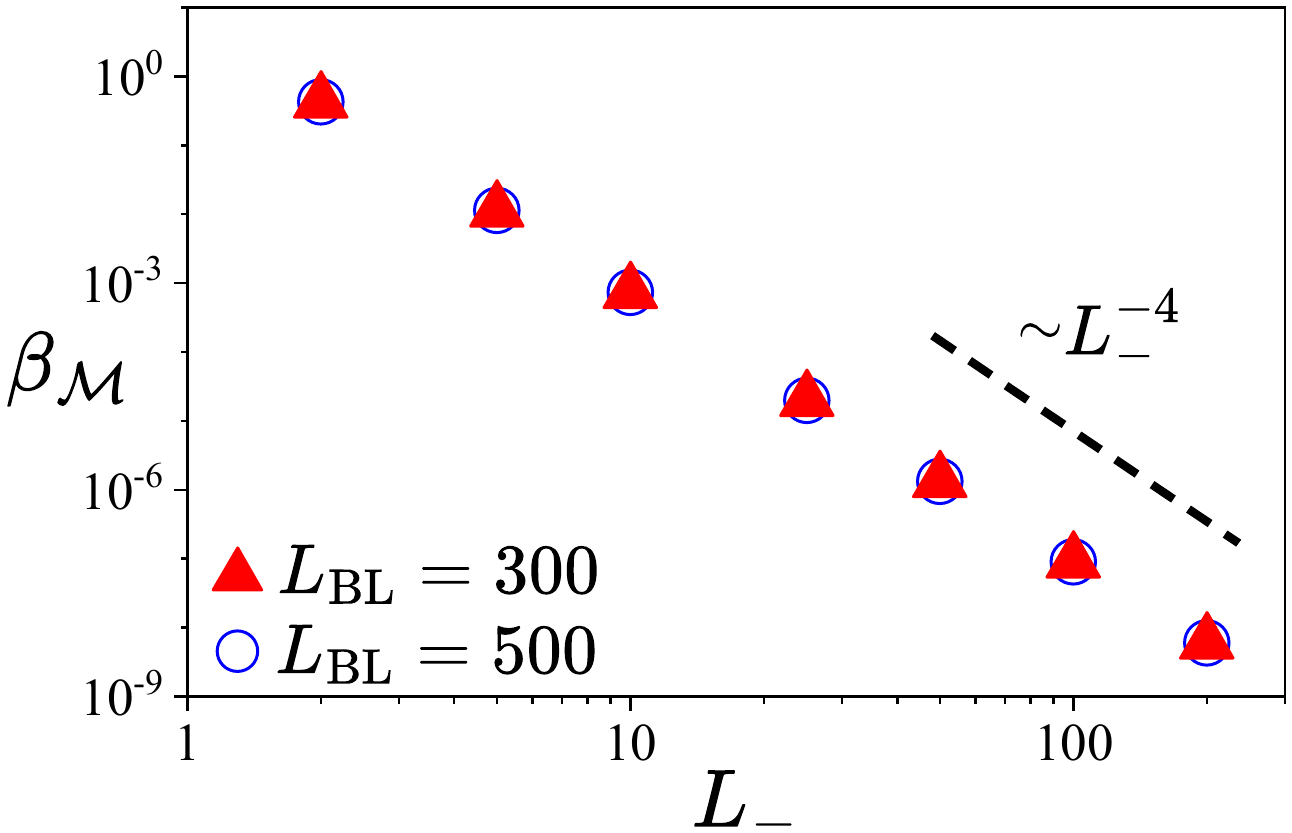}
\caption{Plot of the Mpemba inverse temperature $\beta_{\mathcal M}$ as a function of the position $L_-$ of the soft wall is shown for two different simulation boxes $L_{\text{BL}}=300$ (red solid triangles) and $500$ (blue open circles), and the result is independent of $L_{\text{BL}}$. The dashed line shows excellent agreement with our prediction. We have used $m=4$, a bath temperature $\beta=1$, and the soft wall grows as $V(x)\sim x^6$ ($q=6$) for $-L_\mathrm{BL}<x<-L_-$. Here, we have used $L_+=L_{\text{BR}}=2L_{\text{BL}}$.
}
\label{fig:soft_wall}
\end{center}
\end{figure}

Our results reported here for a single well are almost the same as for a double well in Ref.~\cite{Yue26}.
We also note that our analytical and numerical results agree with Ref.~\cite{kumar2021}. 
%%%%%%%%%%%%%%%%%%%%%%%%%%%%%%
\section{Double well}\label{sec:double}
%%%%%%%%%%%%%%%%%%%%%%

 In this section, we begin by discussing, on general grounds, how the case of a symmetric double-well potential maps onto that of a single-well potential (in subsection~\ref{subsec:symdouble}). 
 Then, in subsection~\ref{sub:symmetric_double}, we use an exactly solvable model (based on a piecewise harmonic potential) to establish, without any approximation, that the normal or inverse Mpemba effects can be observed at the level of $a_3$. 
 We shall also show how this Mpemba effect strongly depends on the presence of confining walls. 
 In subsection~\ref{subsec:asymdouble}, as in~\cite{Yue26}, we explore the low bath temperature limit of $\ell_2$, and we show how the results of the single-well analysis can be used almost directly when the wall is located far away from the center of the potential. 
 In subsection~\ref{sub:proof_absence}, we prove the absence of the Mpemba effect confined in an asymmetric double-well potential with an infinitely large system (this goes beyond~\cite{Yue26} or the previous section because the proof can be used even for $m=2$ in Eq.~\eqref{asymptotic_V(x)}). 
 We also prove the existence of the Mpemba effect when the system is confined between two walls in subsection ~\ref{sub:boundary-Mpemba}.
 This is also a general proof regardless of the detailed properties of the system. Finally, in subsection~\ref{sec:asym_double_well}, we use our qualitative understanding to reverse-engineer a potential that harbors multistage Mpemba effects if we choose a specific potential form.

%In this section, we begin by showing in \ref{subsec:symdouble} that for a symmetric potential, the Mpemba effect can be observed at the level of $a_3$, and under appropriate conditions. Then we show in \ref{subsec:asymdouble} that, in spite of the many differences between the single-well and the double-well potentials, much of the single-well analysis exports to the double-well setting. 
%We then exhibit in \ref{subsec:harmonicdouble} a double potential for which most of the reasoning can be carried out analytically, which allows us to waive some hypotheses (such as the low bath temperature). 
%Finally, in \ref{subsec:multiple} we show how to engineer a potential whose shape leads to multiple Mpemba transition temperatures, thereby confirming our overall physical picture.

%%%%%%%%
\subsection{General features of the symmetric double-well potential}\label{subsec:symdouble}
%%%%%

Previous studies~\cite{kumar2021} have discussed the Mpemba effect in symmetric double-well potentials. 
Interestingly, Ref.~\cite{kumar2020} suggested that to observe the Mpemba effect, asymmetric potentials are needed, whereas Ref.~\cite{kumar2021} predicted its existence in an infinitely extended system without boundaries. 

We now resolve this puzzle. 
As we have discussed in subsection~\ref{subsec:generalsymmetric}, the leading nontrivial contribution arises from the coefficient $a_3$ since $a_2$ vanishes. 
As illustrated in Fig.~\ref{fig:sym_potential}(b), as far as the eigenfunctions $\ell_n$ are concerned, the symmetric double-well potential can be mapped onto an effective asymmetric single-well potential with a wall at the symmetry axis. 
At this stage, we exploit the results we have obtained for a single well, and we immediately conclude that $a_3$ must display non-monotonic behavior and that the Mpemba effect must occur according to the same mechanism as in the single-well case. 
Of course, if two symmetric hard walls are introduced, then the Mpemba effect may disappear if the latter is too close to the well location.

%%%%%%%%%%%
\subsection{Piecewise symmetric harmonic potential}\label{sub:symmetric_double}
%%%%%%%%%%%%%%
Here we use a specific form of the potential that allows us to carry out the calculations exactly, regardless of the initial and bath temperatures. In practice, we define
\begin{equation}\label{V:symmetric}
V(x)=
\begin{cases}
\dfrac{\kappa}{2}(x+1)^2,
& x<-1/2, \\[4pt]
V_\text{M}-\dfrac{\kappa}{2}x^2,
& -1/2\le x\le 1/2, \\[6pt]
\dfrac{\kappa}{2}(x-1)^2,
& x>1/2 ,
\end{cases}
\end{equation}
The continuity of $V$ and $V'$ at $\pm 1/2$ requires $V_\text{M}=\kappa/4$. We have demonstrated in subsection~\ref{subsec:generalsymmetric} and illustrated in Fig.~\ref{fig:sym_potential} that the Mpemba effect in a symmetric double-well potential can be understood through the analysis of a single-well potential in the presence of a hard wall centered at the top of the barrier. 
%Therefore, the Mpemba temperature $\beta_{\mathcal M}$ exhibits the same behavior as that described by Eq.~\eqref{eq:transition} when we analyze the single-well potential. 
This shows that the internal double-well structure and the accompanying metastability are not the primary driver of the Mpemba effect. 
Instead, the presence of a steep wall is essential.

%%%%%%%%%%%%%
While our asymptotic approach, as developed in Sec.~\ref{sec:single}, aided by Appendix~\ref{app:m2}, could be adapted to a potential whose branches grow quadratically (with $m=2$) at infinity, we now perform an exact analysis of the excited states for a single-well potential with a hard wall at $x=0$. 
The set of equations satisfied by $\ell_n^\text{half}$ is
\begin{equation}
    \begin{cases}
    T\p_x^2\ell_n^{\text{half}}-\kappa(x-1)\p_x\ell_n^{\text{half}}=-\lambda_n\ell_n^\text{half},&x>1/2,\\[6pt]
    T\p_x^2\ell_n^{\text{half}}+\kappa x\p_x\ell_n^{\text{half}}=-\lambda_n\ell_n^\text{half},&0<x<1/2,
    \end{cases}
\end{equation}
where we impose three boundary conditions that are the continuity of $\ell_n^\text{half}$ and $\p_x\ell_n^{\text{half}}$ at $x=1/2$, and $\p_x\ell_n^{\text{half}}(0)=0$ at $x=0$. 
In terms of $\varphi_n(x)=\ell_n^\text{half}(x)\sqrt{\pi^{\text{half}}(x,\beta)}$ the eigenvalue equation becomes
\begin{equation}
\beta^{-1}\frac{\dd^{2}\varphi_{n}(x)}{\dd x^{2}}-U_\mathrm{eff}(x)\varphi_{n}(x)
=-\lambda_{n}\varphi_{n}(x).
\end{equation}
where the effective potential defined in Eq.~\eqref{U(x)} reads
\begin{align}
U_\mathrm{eff}(x)=
\begin{cases}
    \frac{\beta \kappa^{2}}{4}x^{2}+\frac{\kappa}{2}, & 0<x<1/2,\\[6pt]
    \frac{\beta \kappa^{2}}{4}(x-1)^{2}-\frac{\kappa}{2}, & x>1/2.
\end{cases}
\end{align}
The solution to these second-order differential equations, known as the Weber equations, is the special parabolic cylinder functions $D_\nu(x)$ recalled in Appendix~\ref{app:Weber_Function}.  
Introducing the index
\begin{align}
\nu_n:=\frac{\lambda_n}{\kappa},
\end{align}
and requiring the decay of $\varphi_n$ to zero at infinity, we can write the solution in the $x>1/2$ region as
\begin{equation}
   \varphi_n(x)=
    A_{n} D_{\nu_n}(\sqrt{\beta \kappa}(x-1)),
\end{equation}
and in the $0<x<1/2$ region as
\begin{equation}
   \varphi_n(x)= 
    B_{n} D_{\nu_n-1}(\sqrt{\beta \kappa}x)+C_{n}D_{\nu_n-1}(-\sqrt{\beta \kappa}x)\,.
\end{equation}
The vanishing of $\varphi'$ at $x=0$ imposes $B_{n}=C_{n}$, and the continuity of $\varphi_n$ and $\varphi_n'$ at $x=1/2$ leads to a $2\times 2$ linear system for the unknowns $A$ and $B$ which admits nonzero solutions  only if the corresponding determinant vanishes:
\begin{equation}
\frac{D'_{\nu_n}(-\sqrt{\beta \kappa}/2)}{D_{\nu_n}(-\sqrt{\beta \kappa}/2)}
=
\frac{D'_{\nu_n-1}(\sqrt{\beta \kappa}/2) -D'_{\nu_n-1}(-\sqrt{\beta \kappa}/2)}
     {D_{\nu_n-1}(\sqrt{\beta \kappa}/2) + D_{\nu_n-1}(-\sqrt{\beta \kappa}/2)},
\label{eq:matching_condition}
\end{equation}
where the prime refers to $D'_\nu(x):=\p_x D_\nu(x)$. 
This is a transcendental equation for $\nu_n=\lambda_n/\kappa$ that eventually determines a discrete set of eigenvalues $\lambda_n$. Once $\lambda_n$ is obtained, the normalization condition
\begin{equation}
\int_{0}^{\infty} \dd x\, \varphi_n^2(x) = 1.
\end{equation}
fixes the remaining constant, and the overlap coefficient for a single well and a hard wall is eventually obtained as
\begin{equation}
    a_n^\text{half} =\int_{0}^{\infty}\dd x\,\pi^{\text{half}}(x,\beta_{i})  \ell^{\text{half}}_n(x,\beta).
\end{equation}
Using the discussion in subsection~\ref{subsec:generalsymmetric}, we now have access to the overlap coefficients for the full symmetric double-well potential that are given by $a_{2n}=0$ and $a_{2n-1}=a_n^\text{half}$.

%%%%%%%%%%%%%%%%%%%%
\subsubsection{Inverse Mpemba effect}
%%%%%%%%%%%%%%%%

%\subsubsection{Inverse Mpemba effect}

To obtain analytical insight, we consider a special choice of parameter for which the eigenvalue can be found exactly, namely, for which the matching condition in Eq.~\eqref{eq:matching_condition} is automatically fulfilled. It turns out that imposing $\beta \kappa=4$ leads to $\nu_2=1$ and thus to $\lambda_{2}= \kappa$ for the single-well potential with a hard wall at $x=0$. For this particular choice,  the parabolic cylinder functions reduce to elementary functions, and then the solution simplifies considerably. While this is a special point in parameter space, we believe it is as good a representative as any other value.

The corresponding eigenfunction for the full symmetric double-well potential has two nodes and therefore corresponds to the third eigenmode, which has even parity. The corresponding left eigenfunction is given by
\begin{equation}
\ell_{3}(x)=
\begin{cases}
A_2 \sqrt{\frac{\beta \kappa}{2}}(x+1), & x<-1/2,\\[4pt]
B_2 \, \ee^{-\frac{\beta \kappa}{2}x^2}, & |x|\le 1/2,\\[4pt]
- A_2 \sqrt{\frac{\beta \kappa}{2}}(x-1), & x>1/2,
\end{cases}
\end{equation}
The continuity condition at $x=1/2$ yields
\begin{equation}
B_2 = \frac{A_2}{2}\sqrt{\frac{\beta \kappa}{2}} \, \ee^{\beta \kappa/8}.
\end{equation}
Finally, we can determine $A_{2}$ by the normalization condition 
\begin{equation}\label{eq:normalization}
    \int_{-\infty}^{\infty} \dd x\, \pi(x,\beta) \ell_3^2(x) = 1,
\end{equation}
which leads to
\begin{equation}\label{eq:A_3}
    A_{2} = \sqrt{Z(\beta)}\left[-\frac{1}{2\sqrt{\ee}} + \frac{\sqrt{2\pi}}{4} + \frac{\sqrt{2\pi}}{2} \text{erf}\left(\frac{1}{\sqrt{2}}\right)\right]^{-\frac{1}{2}}.
\end{equation}
The overlap coefficient $a_3$ can then be computed analytically as
\begin{equation}\label{eq:a_3}
a_3 =
\frac{A_2 \sqrt{2}}{Z(\beta_i)}
\left[
\frac{1}{2} \ee^{\frac{2-\beta_i \kappa}{4}}
M\!\left(\frac{1}{2},\frac{3}{2},\frac{\beta_i \kappa-4}{8}\right)
-\frac{2}{\beta_i \kappa} \ee^{-\frac{\beta_i \kappa}{8}}
\right],
\end{equation}
where $M(a,b,z)$ is Kummer $M$-function defined as~\cite{Abramowitz}:
\begin{equation}\label{def:M}
    M(a,b,z):= {}_1F_1(a;b;z)
    :=\sum_{k=0}^{\infty}\frac{(a)_k}{(b)_k}\frac{z^k}{k!},
\end{equation}
where $(q)_k:=\Gamma(q+k)/\Gamma(q)$  with $\Gamma(x):=\int_0^\infty \dd t\,  t^{x-1}\ee^{-t}$ is the Pochhammer symbol.
\begin{comment}
, and Kummer's U function defined as
\begin{align}\label{U_connection}
    U(a,b,z)
    :&= \frac{\Gamma(1-b)}{\Gamma(a-b+1)} M(a,b,z) \notag\\
    &\quad + \frac{\Gamma(b-1)}{\Gamma(a)} z^{1-b} M(a-b+1,2-b,z) .
\end{align}
\end{comment}
We also note that $Z(\beta_i)$ can be explicitly written as
\begin{equation}
    Z(\beta_i) = \sqrt{\frac{2\pi}{\beta_i \kappa}} \left[ \text{erfc}\left(-\sqrt{\frac{\beta_i \kappa}{8}}\right) + \ee^{-\frac{\beta_i \kappa}{4}} \text{erfi}\left(\sqrt{\frac{\beta_i \kappa}{8}}\right) \right] ,
\end{equation}
where \(\text{erf}(x):=\frac{2}{\sqrt{\pi}}\int_0^x \dd t\,  \ee^{-t^2}\), $\text{erfi}(x):=\frac{2}{\sqrt{\pi}}\int_0^x \dd t\,  \ee^{t^2}$, and $\text{erfc}(x):=\frac{2}{\sqrt{\pi}}\int_x^\infty \dd t\,  \ee^{-t^2}$.

%\textcolor{red}{Is $\beta$ in Eq.~\eqref{eq:A_3} $\beta_i$?}\textcolor{orange}{YL: $\beta$, because $A_{3}$ is the coefficient of $l_{3}$, which depends on bath temperature.}

This expression explicitly shows that $a_3$ depends non-monotonically on the initial temperature $\beta_i$. This establishes the existence of the inverse Mpemba effect. 

%%%%%%%%%%%%%%%%%

To illustrate our analytical result, we numerically compute the eigenvalues and eigenfunctions of the Fokker-Planck operator using the finite-difference method. We discretize the horizontal axis into $N=10^4$ lattice points. The evolution operator $\mathbb{W}$ can be represented by an $N\times N$ matrix, and the eigenvalues and eigenvectors can be obtained by numerical diagonalization of the matrix.

%%%%%%%%%%%%%%%%%%%%%
\begin{figure}[t]
\begin{center}
\includegraphics[width=1\linewidth]{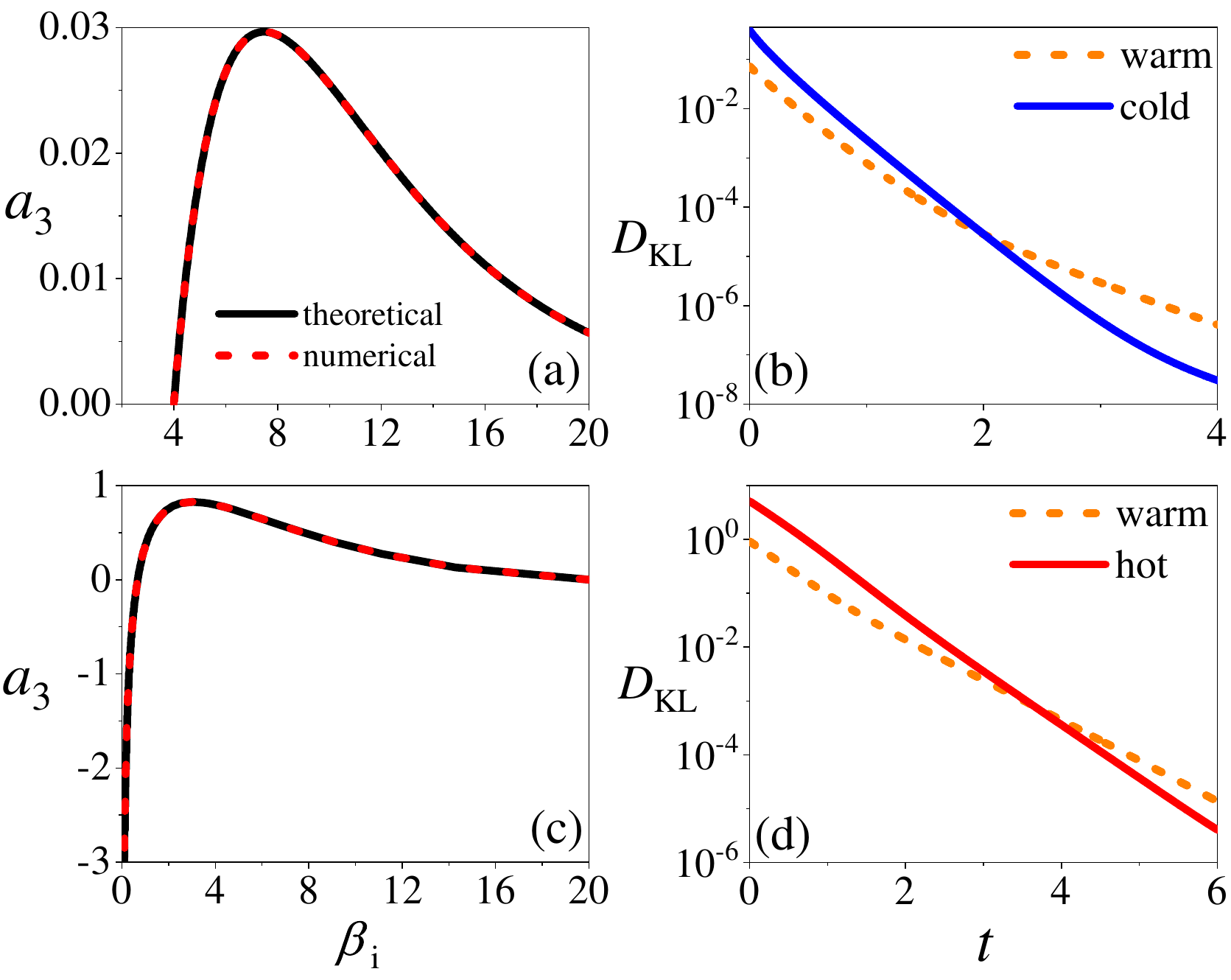}
\caption{(a) The overlap coefficient $a_{3}$ as a function of $\beta_{i}$ for the symmetric double-well potential constructed by a piecewise harmonic potential, which exhibits a non-monotonic behavior, indicating the existence of the (inverse) Mpemba effect. The analytical result (black solid line) is consistent with the numerical result (red dashed line). (b) The time evolution of the KL divergence for cold initial temperature ($\beta_{i}=20$, blue solid line) and warm initial temperature ($\beta_{i}=8$, orange solid line), which cross each other, thereby confirming the (inverse) Mpemba effect. The inverse bath temperature is set to $\beta=4$ for both panels (a) and (b). (c) The same as (a) but $\beta=20$. (d) The time evolution of the KL divergence for hot initial temperature ($\beta_{i}=1$, red solid line) and warm initial temperature ($\beta_{i}=4$, orange solid line), which also cross each other, thereby confirming the Mpemba effect. The inverse bath temperature is set to $\beta=20$ for both panels (c) and (d). For all panels, the simulation box is set to $L_{\text{BL}}=L_{\text{BR}}=100$ and $\kappa=1$.
}
\label{fig:sym_double_well}
\end{center}
\end{figure}
%%%%%%%%%%%%%%%%%%%%

As shown in Fig.~\ref{fig:sym_double_well}(a), the overlap coefficient $a_{3}$ exhibits a non-monotonic behavior with respect to the initial temperature, with a maximum at around $\beta_i \kappa=8$.
This demonstrates the existence of the Mpemba effect.
Moreover, the analytical result is in good agreement with our numerics, as shown in Fig.~\ref{fig:sym_double_well}.
Note, however, that for the specific value we focus on, $\beta\kappa=4$, we only observe an inverse Mpemba effect, i.e., an anomalous relaxation in the heating process.

Then, we numerically integrate the time evolution of the probability density function for different initial temperatures using a finite-difference method, and we use the KL divergence to quantify the distance between the probability density function at time $t$ and the final equilibrium distribution. 
As shown in Fig.~\ref{fig:sym_double_well}(b), we observe that the KL divergence for the cold initial condition crosses the KL divergence for the warm initial condition, which confirms the existence of the Mpemba effect. 
The exactly solvable potential we considered here demonstrates, without approximation or asymptotics, the existence of the Mpemba effect in the symmetric double-well potential.

\begin{comment}
Since we cannot use Eq.~\eqref{beta_k=4} to probe a quench to a lower bath temperature, the corresponding solution is expressed in a complicated form, not like Eq.~\eqref{eq:a_3}, corresponding to the 2D case~\cite{Hayakawa2026}. 
In this case, it is possible to prove the absence of $\beta_\mathcal{M}$ in an infinitely large system similar to an analytic potential satisfying Eq.~\eqref{asymptotic_V(x)} with $m>2$.
%\textcolor{red}{To Yue: Please write your argument without the condition Eq.~\eqref{beta_k=4} for a quench into a lower temperature.}
\end{comment}

%%%%%%%%%%%%
\begin{figure}[t]
\begin{center}
\includegraphics[width=1\linewidth]{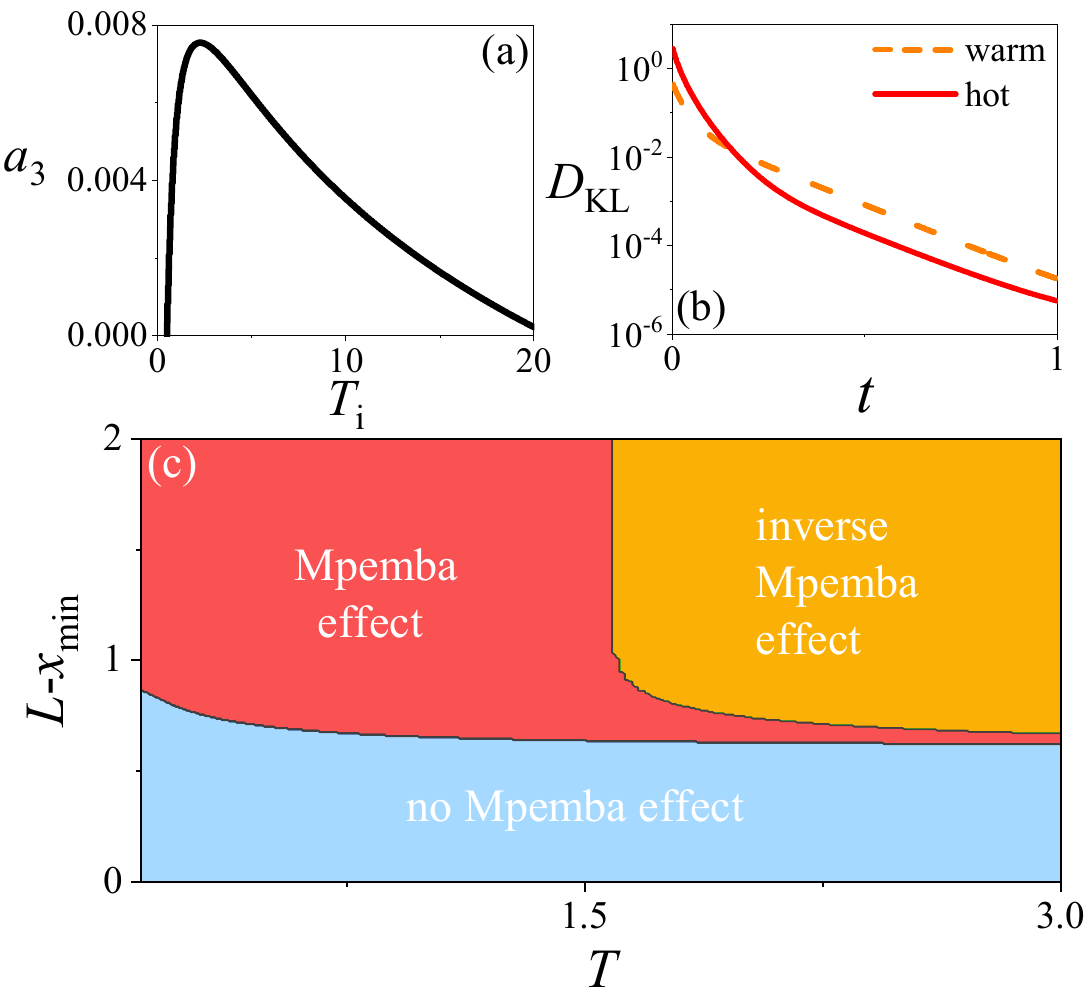}
\caption{The Mpemba effect in the symmetric double-well potential $V(x)=(x^{2}-1)^{2}$. (a) The overlap coefficient $a_{3}$ as a function of initial temperature $T_{i}$. 
%\textcolor{red}{To Yue: Figure plots $a_3$ against the bath temperature $T$. Please correct the sentence or the figure.}
The observed non-monotonic behavior establishes the existence of the Mpemba effect. (b) Time evolution of the KL divergence for a hot initial condition ($T_{i}=4$, red solid line) and a warm initial condition ($T_{i}=10$, orange solid line). The crossing of the two curves dynamically confirms the anomaly. For panels (a) and (b), the bath temperature is fixed at $T=0.5$. (c) Phase diagram for the system confined by symmetric walls at $\pm L$, plotted against the bath temperature $T$ and the relative wall distance $L - x_{\text{min}}$ (where $x_{\text{min}}=-1$ is the potential minimum). The colored regions delineate the parameter regimes harboring the normal Mpemba effect (red), the inverse Mpemba effect (orange), and no Mpemba effect (blue). The simulation box is set to $L_{\text{BL}}=L_{\text{BR}}=100$ for (a) and (b), and $L_{\text{BL}}=L_{\text{BR}}=L$ for (c). 
}
\label{fig:sym_double_well_x4}
\end{center}
\end{figure}
%%%%%%%%%%%%

%%%%%%%%%%%%%%%%%%%
\subsubsection{Mpemba effect for a quench to a lower temperature}
%%%%%%%%%%%%%%%

%\subsubsection{Mpemba effect for a quench into a lower temperature}

If we abandon the condition $\beta \kappa = 4$ and work at an arbitrary bath temperature,  no closed-form solution is readily available, and the eigenvalue problem in Eq.~\eqref{eq:matching_condition} must be solved numerically.
In that case, we numerically evaluate the overlap coefficient $a_3$ using the general expressions derived in subsection~\ref{sub:symmetric_double}. 
As shown in Fig.~\ref{fig:sym_double_well}(c), we observe a clear Mpemba effect: the coefficient $a_3$ for a higher initial temperature can indeed be smaller than that for a lower initial temperature. 
In Fig.~\ref{fig:sym_double_well}(d), we also confirm the Mpemba effect by directly simulating the Fokker-Planck dynamics and observing the crossing of the KL divergence trajectories for a hot initial condition and a warm initial condition.
Moreover, as shown in Fig.~\ref{fig:sym_double_well}(c), there exists a temperature $T_i$ at which $a_3$ vanishes, which is known as a strong Mpemba effect. 

To further probe the scope of these results, we also analyze the smooth quartic potential $V(x)=(x^2-1)^2$. Numerical evaluation of $a_3$ again reveals non-monotonic dependence on the initial temperature. In
Fig.~\ref{fig:sym_double_well_x4}(a), we plot the calculated $a_3$ against the initial temperature $T_i$, revealing a clear non-monotonic behavior that directly signals the presence of the Mpemba effect. To capture the full relaxation dynamics, we numerically solve the time evolution of the probability density function by the finite-difference method and quantify its distance to thermal equilibrium via the KL divergence. 
As depicted in Fig.~\ref{fig:sym_double_well_x4}(b), the KL divergence trajectory for a hot initial state explicitly crosses that of a warm initial state, providing unambiguous dynamical evidence of the Mpemba effect in this symmetric double-well system. 

Finally, we consider the system confined by symmetric walls at $\pm L$. In Fig.~\ref{fig:sym_double_well_x4}(c), the phase diagram shows that both normal and inverse Mpemba effects appear when the walls are sufficiently far from the potential minima. 
In contrast, when the walls are close to the minima, the effect disappears. 
This behavior is consistent with the interpretation that the effective asymmetry---the one induced by the mapping to a single-well system---is suppressed by a strong confinement.
Note that the phase diagram for $V(x)=(x^2-1)^2$ is completely different from that for Eq.~\eqref{V:symmetric}.
Thus, we can observe the Mpemba effect for $T=1/4$ in Fig.~\ref{fig:sym_double_well_x4}(c).

%
% \textcolor{red}{To Yue: Please check the consistency between Fig. \ref{fig:sym_double_well} for $\beta=4$ ($T=1/4$) and Fig. \ref{fig:sym_double_well_x4}.
% I believe that the label $T$ for the horizontal axis in Fig. \ref{fig:sym_double_well_x4} should be $\beta$.}

Our results in this subsection conclude that the Mpemba effect can be observed even in a symmetric double-well potential if the system is sufficiently large, as in Ref.~\cite{kumar2021}.

%%%%%%%%%%%%%%%%%%%%%%
\subsection{Asymmetric double-well potential}\label{subsec:asymdouble}
%%%%%%%%%%%%%%%%%%%%%

In this subsection devoted to the study of a generic asymmetric double-well potential, we briefly revisit the arguments sketched in~\cite{Yue26}. 
A key ingredient of our analysis is that the eigenfunction $\ell_2$ has a step-like behavior in the low bath temperature regime. 
We have already known from subsection~\ref{subsec:firstexcited} that $\ell_2$ reaches plateau values at infinity (or if the walls are sufficiently far away), so we only need to explain why the width of the step shrinks as the bath temperature becomes low. This behavior at infinity holds as long as the potential grows faster than $x^2$ at infinity, but Appendix~\ref{app:m2} shows how to deal with a potential growing quadratically at infinity (see the corresponding behavior for $\beta_{\mathcal{M}}$ as a function of $L_{-}$ in Fig.~\ref{fig:m=2}).

The physics behind this mechanism strongly differs from that of the single well explored in subsection~\ref{subsec:l2single}. 
In a double-well setting, there exists a very slow timescale associated with the crossing of the energy barrier between the two wells. 
According to the Kramers' theory, in a low bath temperature regime, $\lambda_2\propto \ee^{-\beta\Delta V_{b}}$, where $\Delta V_{b}$ is the potential barrier from the least stable well. 
The behavior of $\ell_2$ has been explored by mathematical techniques~\cite{matkowsky1977exit,bovier2004metastability}, but it is not easy to find a simple physicist's argument in the literature. 
This is what we now provide here. 
From Eq.~\eqref{eq:defl2} the Fokker-Planck equation for $\ell_2(x)$ can be simplified to $\ell_2''(x)-\beta V'(x)\ell_2'(x)\simeq 0$ if $\ell_2''\gg\lambda_2\ell_2$.
For this approximation to be true, the condition $\ell_2/w^2\gg\lambda_2\ell_2$ should be satisfied, where $w$ is the width of $\ell_2$.
In other words, we must be in a regime such that $w^{-2}\gg\lambda_2$. 
Assuming this condition is fulfilled, we integrate this equation twice, and we obtain the general solution for $\ell_2$:
\begin{equation}
    \ell_2(x)=C_0+A\int_{x^{*}}^{x}\dd y\,\ee^{\beta V(y)},
\end{equation}
where $x^{*}$ is the position of the barrier saddle point, and $A$ and $C_0$ are two constants to be determined. To evaluate this integral, we apply a saddle point expansion around the local maximum $x^{*}$ of $V$, which leads to
\begin{equation}
    V(y)\simeq V(x^{*})-\frac{1}{2}|V''(x^{*})|(y-x^{*})^{2},
\end{equation}
where we expand the potential to the second order, and use the fact that $V'(x^{*})=0$ and $V''(x^{*})<0$. Substituting this approximation into the expression of $\ell_2(x)$, it takes the form of an error function
\begin{equation}
    \ell_2(x) \simeq C_0 + A \cdot \text{erf} \left[ \sqrt{\frac{\beta |V''(x^*)|}{2}} (x - x^*) \right],
\end{equation}
where $A$ absorbs the constant from the integration. The error function approaches $-1$ in the left well ($x \to -L_-$ with $L_-\gg 1$) and $+1$ in the right well ($x \to L_+\gg 1$). 
%\textcolor{red}{Quite confusing. Do you consider the case of $L_\pm \to \infty?$}
This means $\ell_2(x)$ has constant flat values in the two wells: $\ell_2(-\infty) = C_0 - A$ and $\ell_2(+\infty) = C_0 + A$ in the limit $L_\pm\to \infty$. 
These constants can be determined by the orthogonality condition $\int \dd x\, \pi(x,\beta)\ell_2(x)=0$. 
The equilibrium distribution is highly localized around the two minima of the potential, and the ratio of the population in the two wells is given by $\ee^{-\beta\Delta V}$, 
which leads to $p_{-}\ell_2(-L_-)+p_{+}\ell_2(L_+)\simeq 0$ for sufficiently large $L_\pm$. 
Here, $\Delta V$ is the potential difference between the two wells.
If we choose the normalization convention where $\ell_2$ in the left well is set to 1, we have 
\begin{equation}
    \begin{split}
    \ell_2(x) \simeq  &- \frac{1 + \ee^{-\beta \Delta V}}{2} \text{erf} \left[ \sqrt{\frac{\beta |V''(x^*)|}{2}} (x - x^*) \right]\\
    &+\frac{1 - \ee^{-\beta \Delta V}}{2}.
    \end{split}
\end{equation}
The characteristic width of the error function is $w(T)\sim 1/\sqrt{\beta|V''(x^*)|}$, which vanishes as the temperature goes to zero. The self-consistent condition for our calculation to be valid is therefore 
\begin{equation}
    w^{-2}\gg \lambda_2\Longleftrightarrow \beta |V''(x^*)|\gg \lambda_2
\end{equation}
which is well verified in the low bath temperature limit, since $\lambda_2\propto\ee^{-\beta\Delta V_{b}}$.

Therefore, at low bath temperature, $\ell_2(x)$ also becomes a step function, and its derivative can be approximated as a downward delta peak localized at the potential barrier. 
In the presence of walls located sufficiently far away from $x^*$, we see from the single-well derivation in subsection~\ref{subsec:derivationsinglewall} that whatever internal structure the potential exhibits (be it single or double well), what matters is the asymptotic behavior in the vicinity of the wall, and our analytic conclusions from the single well directly export to the double-well configuration. 
Note, however, that the double-well potential is associated with a very long time scale $\lambda_2^{-1}$ in the low bath temperature regime, which makes the concrete observation of the Mpemba effect easier than in the single-well setting. 

We now show that there is no Mpemba effect in the absence of walls for a cold enough initial condition, regardless of the bath temperature.

%%%%%%%%%%%
\subsection{General proof for the absence of the Mpemba effect for a cold initial condition in the absence of confining walls}
\label{sub:proof_absence}
%%%%%%%%%%%%%
\begin{comment}
As demonstrated in Ref.~\cite{Yue26}, Eq.~\eqref{eq:transition} can be used even for an asymmetric double-well potential with $m>2$ if we begin with a low initial temperature.
However, this argument cannot be directly applied to $m\le 2$, e.g., for a connected harmonic-antiharmonic-harmonic potential with $m=2$. 
\end{comment}
In this subsection, we briefly prove the absence of the Mpemba effect for a general unbounded asymmetric potential, i.e., $L_\pm\to \infty$ if we begin with an initial condition at sufficiently low temperature.
This proof can be used even for $m=2$ in Eq.~\eqref{asymptotic_V(x)}, which was excluded in the previous argument.

Let us consider Eq.~\eqref{eq:derivativea2} for an asymmetric double-well potential. We denote by $x^*$ the location of the barrier, and by $x_\pm$ the locations of the deep/shallow wells ($x_-<x^*<x_+$). Since $\pi(x,\beta_i)$ is positive, the sign of $\partial a_2/\partial \beta_i$ is determined by the sign of $\ell_2(x)$ and $U_i-V(x)$. As demonstrated in Appendix~\ref{Fixed_sign_l2}, the node $x_0$ of $\ell_2(x)$ must lie in between $x_-$ and $x_+$. Then, $\ell_2(x)$ has a definite sign in the left/right well with opposite sign. If we choose $\beta_i$ large enough, then, because $x_0<x_+$, over the region where $\pi$ is non-zero, $\ell_2$ is always negative, regardless of the width of $\ell_2$, that is, regardless of the bath temperature.
%\begin{comment}
Thanks to Eq.~\eqref{phi_r_ell}, $\ell_2(x)$ includes the factor $e^{-\beta V(x)/2}$.
Thus, $\ell_2(x)$ must be localized in the right well at a sufficiently low temperature (large $\beta$).
Therefore, we conclude that $\ell_2(x)$ has a fixed sign.
%\end{comment}
Since $U_i-V(x)$ also has a fixed sign for sufficiently low initial temperature, as demonstrated in Appendix~\ref{Sign_DV}.
Then we conclude that $\partial a_2/\partial \beta_i$ cannot vanish when the initial temperature is small enough.
This prohibits the solution of $\partial a_2/\partial \beta_i=0$, which determines $\beta_\mathcal{M}$. Thus, there is no Mpemba effect for an asymmetric double-well potential if the system is unbounded, i.e., $L_\pm\to \infty$.
The argument in this subsection can be used even for a connected harmonic-antiharmonic-harmonic potential, characterized by $m=2$.

It should be noted that this subsection only indicates the absence of $\beta_\mathcal{M}$ as the solution of $\partial a_2/\partial \beta_i=0$.
Interestingly, the absence of $\beta_\mathcal{M}$ does not guarantee the absence of any crossing of observables, such as $D_\mathrm{KL}(p||\pi)$.
Indeed, if we consider the behavior of a particle confined in a weakly asymmetric potential without walls, we find double crossings of $D_\mathrm{KL}(p||\pi)$, in which $a_2$ is monotonic against $\beta_i$ but $a_3$ has a maximum at a certain $\beta_i$.

To illustrate this picture, we have analyzed a weakly asymmetric double-well potential $V(x)=(x^{2}-1)^{2} + \epsilon x $ with  $\epsilon$ as small as $\epsilon=0.01$ at $T=0.5$.
As in Fig.~\ref{fig:sym_double_well_x4} with $\epsilon=0$, $a_3$ has a single peak against $T_i$ or $\beta_i$, while $a_2$ is a monotonic increasing function of $T_i(>T)$.
In this case, we confirm the existence of the double-crossing of the KL divergence (see Fig.~\ref{fig:fig_DKL_small_k}).
The double crossings (or even number crossings) can not be regarded as the Mpemba effect in the sense of long-time behavior~\cite{Teza25}. Similar multiple crossings have been observed in many previous studies by considering nonequilibrium initial conditions~\cite{Malhotra24,Takada21a,Chatterjee_2023,Chatterjee24}. Here, we illustrate that the double crossing can also be observed for equilibrium initial conditions, which is a consequence of the non-monotonic behavior of $a_3$ and monotonic behavior of $a_2$.

%%%%%%%%%%%
\begin{figure}[t]
\begin{center}
\includegraphics[width=1\linewidth]{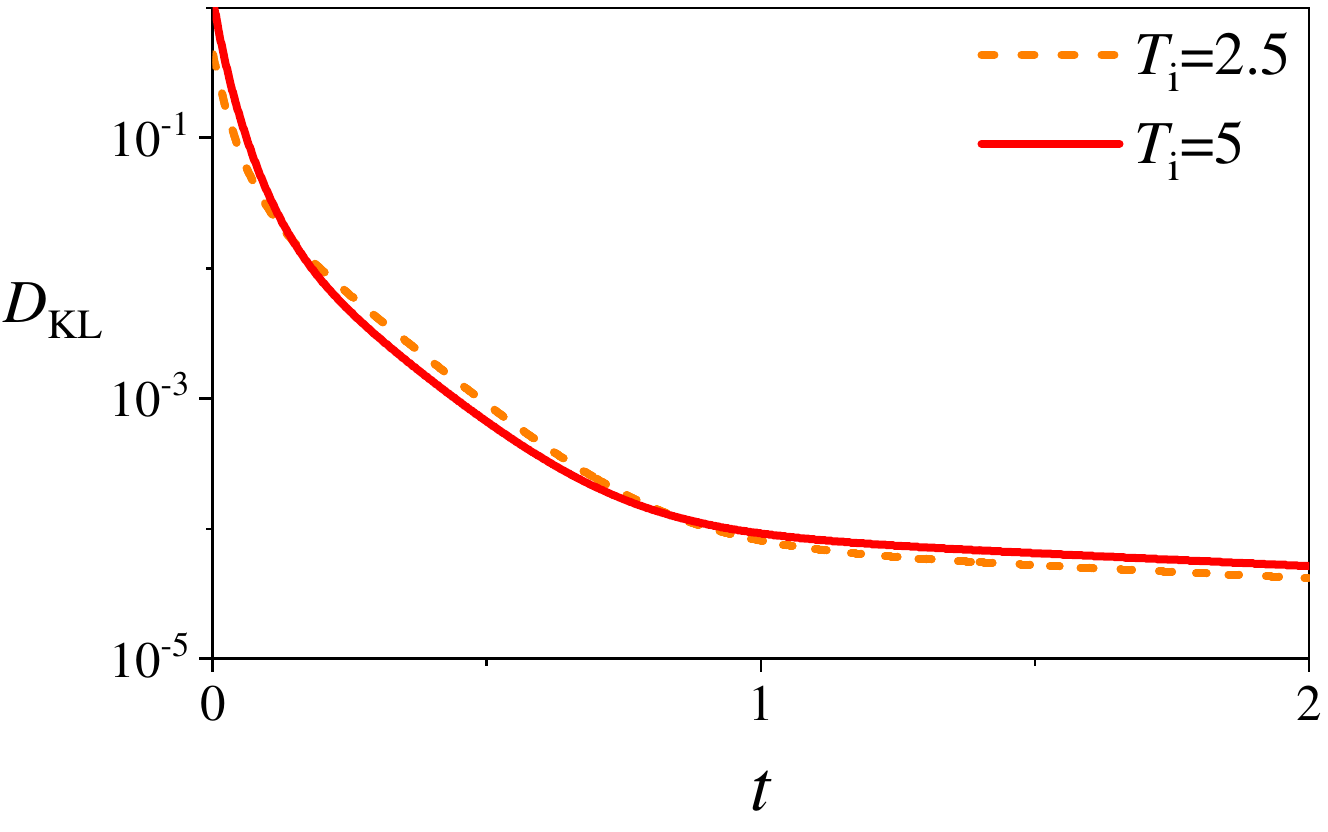}
\caption{The time evolution of $D_\mathrm{KL}(p||\pi)$ for $V(x)=(x^{2}-1)^{2} + 0.01 x $ at $T=0.5$ is shown for different initial temperature $T_{i}=5$ (red) and $T_{i}=2.5$ (orange). The simulation box is set to $L_{\text{BL}}=L_{\text{BR}}=100$.
} 
\label{fig:fig_DKL_small_k}
\end{center}
\end{figure}
%%%%%%%%%%

%%%%%%%%%%%%%%%%
\subsection{Boundary-induced Mpemba effect}\label{sub:boundary-Mpemba}
%%%%%%%%%%%

To guarantee the existence of $\beta_{\mathcal{M}}$, we can analyze the energy difference between the left and right regions, defined as $\delta U(\beta_i):= U_-(\beta_i) - U_+(\beta_i)$. 
In the limit $\beta_i \to \infty$, the populations are perfectly localized at the local minima, yielding $\delta U(\infty) = V(x_-) - V(x_+) > 0$ (assuming the left well is shallower). 
In the limit $\beta_i \to 0$, the left energy $U_-(\beta_i\to 0)$ is bounded by the hard wall at $-L_-$, whereas the right energy $U_+$ diverges to infinity due to the unbounded integration domain on the right side. 
Consequently, $\delta U(0) \to -\infty$. 
Because $\delta U(\beta_i)$ is a continuous function, the intermediate value theorem ensures that it must cross zero at least once. 
This proof (detailed in Appendix~\ref{app:transition_temp}) dictates that a wall on the shallower side guarantees at least one Mpemba temperature, even for $m=2$ in Eq.~\eqref{asymptotic_V(x)}. 
Therefore, we verified that the Mpemba effect observed in the asymmetric potential originates from the influence of the sidewalls.

%%%%%%%%%%%%%%
\subsection{Towards a multistage Mpemba effect?} \label{sec:asym_double_well}
%%%%%%%%%%%%%
The initial energy difference between the two wells $\delta U(\beta_i)$ is not constrained to be monotonic; the energy difference can cross zero multiple times, which can also give rise to as many Mpemba effects as desired, as discussed below. 
We term this phenomenon the multistage Mpemba effect. 
In Ref.~\cite{Yue26}, we demonstrated a scenario of the coexistence of a maximum and a minimum for $a_2$ using an asymmetric double well where the shallow and deep sides scale as $x^4$ and $x^2$, respectively. 
As the initial temperature rises, the population initially shifts to the shallow well but is subsequently forced back into the deep well due to the steeper $x^4$ potential, triggering the first Mpemba temperature. 
A second Mpemba temperature emerges if a steep wall is placed on the deep-well side, eventually reflecting the population towards the shallow side at even higher temperatures.

Extending this logic, we can design a potential exhibiting three Mpemba temperatures
\begin{equation}
V(x) = \begin{cases}
x^4 - 2x^2 & \text{for } x < 0, \\
\dfrac{x^4 + b x^8}{1 + a x^2} - 2x^2 & \text{for } x > 0.
\end{cases}
\end{equation}
where the parameter $b$ is chosen to be sufficiently small so that the deep well is governed by $x^2$ at small $x$ and $x^6$ at large $x$. The mechanism unfolds in three stages. 
First, as shown in Fig.~\ref{fig:sym_double_well_multi}(a), the $x^4$ versus $x^2$ competition drives the population from the deep well to the shallow well and back, yielding the initial Mpemba temperature. 
Second, as the $x^6$ term dominates at higher temperatures (see Fig.~\ref{fig:sym_double_well_multi}(b)), the population is driven back toward the shallow side, producing the second transition. Finally, introducing a confining wall on the shallow side forces a final transfer back into the deep well, generating the third Mpemba temperature. 
As a result, Fig.~\ref{fig:sym_double_well_multi}(c) shows three Mpemba temperatures associated with $a_{2}$. 
In principle, this approach allows us to engineer potentials with an arbitrary number of Mpemba temperatures. This concept is equally applicable to asymmetric single-well systems.

% Multiple Mpemba effects have been reported in many previous papers, such as Ref.~\cite{Takada21a,Chatterjee_2023,Chatterjee24} under nonequilibrium initial conditions.
% However, the significant result in this subsection is that we can clarify the condition for the multiple Mpemba effects under initial equilibrium conditions, as reported in Ref.~\cite{Malhotra24}.

Note that what we have called the multistage Mpemba effect differs from the observation of multiple crossings of the distance function (as illustrated in Fig.~\ref{fig:fig_DKL_small_k}). The double Mpemba effect (i.e., the existence of two Mpemba temperatures) has been numerically observed in Ref.~\cite{Malhotra24} in a double-well potential. Here, we have provided an understanding of the mechanism behind this phenomenon, along with a recipe for engineering potentials with an arbitrary number of such Mpemba temperatures (this applies equally to single- or double-well potentials). Hopefully, the multistage Mpemba effect will not remain a theoretical curiosity; it could be experimentally realized in systems on condition that the potential landscape can be finely tuned according to our recipe.

%%%%%%%%%%%
\begin{figure}[t]
\begin{center}
\includegraphics[width=1\linewidth]{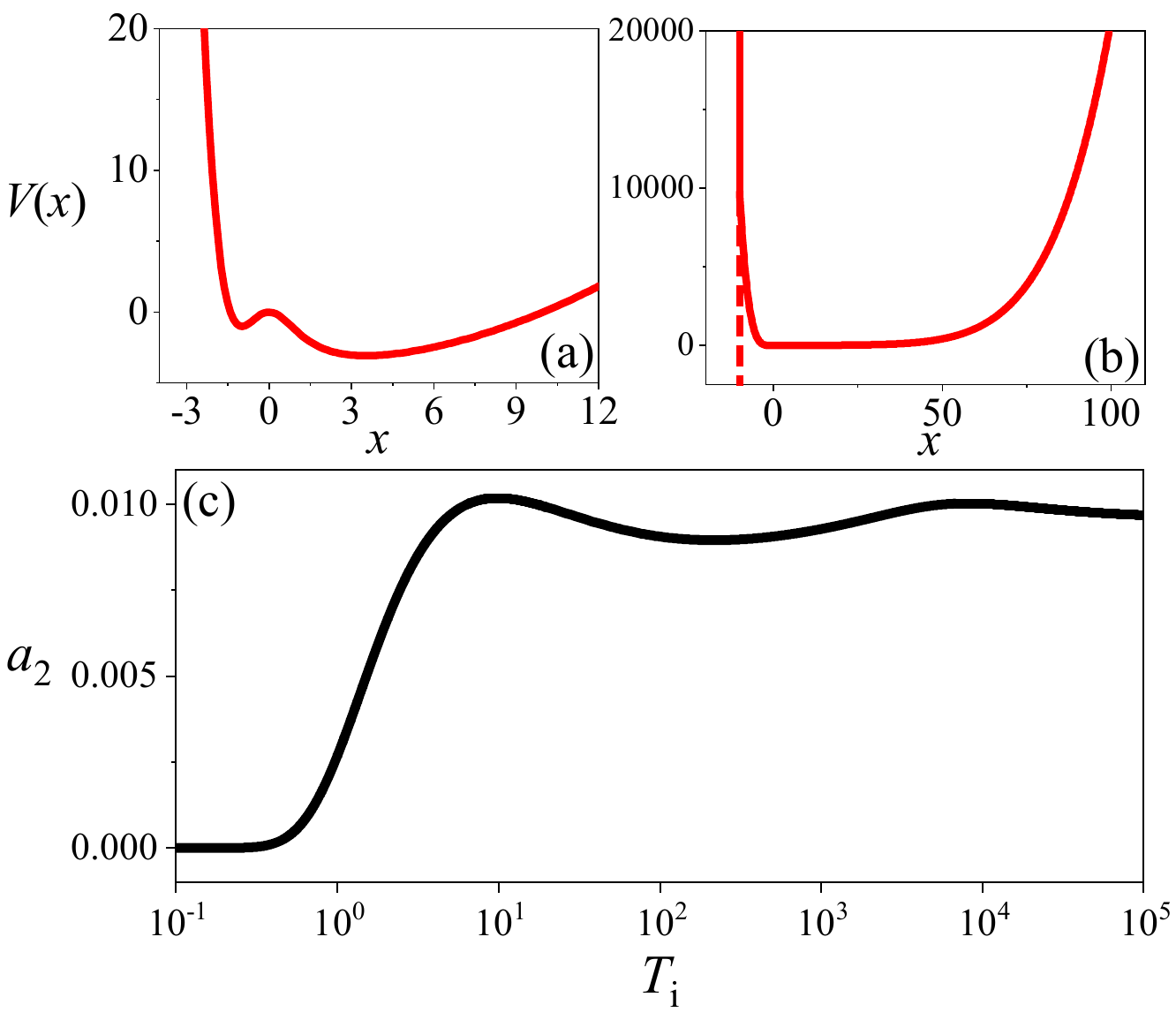}
\caption{Engineered potential and the resulting multistage Mpemba effect. (a) A local magnified view of the constructed asymmetric double-well potential $V(x)$. (b) The global profile of the identical potential over a broader spatial range. (c) The overlap coefficient $a_2$ plotted as a function of the initial temperature $T_i$. Here, the parameters are set to $T=0.1$, $a=0.49$ and $b=10^{-8}$, and the simulation box is set to $L_{\text{BL}}=10$ and $L_{\text{BR}}=200$.}  
\label{fig:sym_double_well_multi}
\end{center}
\end{figure}
%%%%%%%%%%

%%%%%%%%%%%%%
\section{Conclusions and outlook}
\label{sec:conclusion}
%%%%%%%%%%%%%%%%%%%%

In this work, we have provided a comprehensive classification of the conditions required to observe the Mpemba effect in one-dimensional overdamped Langevin systems starting from equilibrium initial conditions. 
Our analysis reveals that the existence of this anomalous relaxation phenomenon is fundamentally tied to the presence of boundaries, whether hard or soft, rather than to the specific internal features of the potential landscape, such as the presence of multiple minima or energy barriers. 

For single-well potentials, we established that the Mpemba effect is guaranteed in asymmetric landscapes, provided at least one wall is present. 
In the case of double-well potentials, the behavior is more nuanced; while asymmetry and boundary conditions remain primary drivers, even symmetric double-wells can exhibit the effect under specific configurations. 
A key theoretical pillar of our findings is that in the low-bath-temperature regime, the derivative of the first excited state eigenfunction, $-\partial_x \ell_2(x)$, effectively behaves as a Dirac delta peak.
This behavior combined with the presence of the wall, which influences the initial population distribution in a non-trivial way, leads to a non-monotonic dependence of the overlap coefficient $a_2$ on the initial temperature, which is a necessary condition for the Mpemba effect.

Furthermore, we demonstrated the ability to engineer potentials that exhibit what we termed a multistage Mpemba effect. 
By carefully constructing asymmetric landscapes where different power-law behaviors (e.g., $x^2, x^4, x^6$) dominate at distinct energy scales, such landscapes can be viewed as a sequence of increasingly steeper walls, and we can force the population to shift back and forth between wells as the initial temperature increases, resulting in multiple zeros for the derivative of the overlap coefficient. 
This systematic approach provides a recipe for designing systems with an arbitrary number of relaxation anomalies.

Several open questions remain for future exploration. While our current study focuses on the necessary condition (the non-monotonicity of $a_2$), a full treatment of the sufficient condition requires accounting for higher-order modes, specifically $a_3$ as in Ref.~\cite{Hayakawa2026}. 
Additionally, while we have shown that the qualitative behavior of systems with soft walls mirrors that of systems with hard walls, a more rigorous mapping between these two regimes would be beneficial. 
Finally, extending this wall-induced framework to higher-dimensional systems and quantum mechanical relaxation remains a promising frontier for further research.

%%%%%%%%%%%%
%\section{Outlook}\label{Sec:outlook}
%%%%%%%%%%%

%\textcolor{red}{To Yue: You have not cited Refs.~\cite{walker2021anomalous,morsch1979one,van1992stochastic}.
%I believe that Ref.~\cite{walker2021anomalous} should be cited in a proper place.
%}

\begin{acknowledgements}
We appreciate the useful exchange of communications with John Bechhoefer, Siddharth Sane, Marija Vucelja and Apurba Biswas. 
HH also thanks Satoshi Takada for fruitful discussions and was partially supported by JSPS KAKENHI Grant No. JP26K06960.
FvW acknowledges the financial support of the ANR THEMA AAPG2020 grant. 
TVV was supported by JSPS KAKENHI Grant Nos.~JP23K13032, JP26K00019, JP26K00022, and JP26H02015. 
YL gratefully acknowledges the Yukawa Research Fellow, co-sponsored by the YITP and the Yukawa Memorial Foundation. 
\end{acknowledgements}

\section*{Data availability}
The data are not publicly available upon publication. The data are available from the authors upon reasonable request.

% \newpage

%%%%%%%%%%%%%%%%%%%%%%%%
\begin{appendix}

%%%%%%%%%%%%%%%
\section{The $m=2$ case}\label{app:m2}
In the main text, we have focused on the $m>2$ case. This appendix briefly explains what needs to be modified when $m=2$ and when there exists a hard wall. We are now in the situation where $V(x)$ grows quadratically with $x$ for large $x$, until a hard wall at $\pm L_\pm$ large enough is felt. For a double well potential, in the low bath temperature regime, we know that $\lambda_2\sim\ee^{-\beta \Delta V_b}$, which can become arbitrarily small as the bath temperature is decreased. The differential equation satisfied by $\ell_2$ reads 
\begin{equation}
T\ell_2''(x)-    V'(x) \ell_2'(x) =-\lambda_2 \ell_2(x),
\end{equation}
As we shall a posteriori verify, $\ell_2$ grows as a power law with positive exponent for large $x$, and thus the diffusive contribution can be neglected. Using $V'(x)\sim x$, we obtain, say for $L_+\gg x\gg x_+$
\begin{equation}
    \frac{\ell_{2}'(x)}{\ell_2(x)} \sim \frac{\lambda_2}{x}.
\end{equation}
Integrating this equation, we have
\begin{equation} 
    \ell_2(x)\simeq A_+ x^{\lambda_2},
\end{equation} 
which is indeed a power-law growth. Since $\lambda_2$ can be made arbitrarily small by decreasing the bath temperature, we can write that as $\lambda_2\to 0$ but at fixed $x$,
\begin{equation}
    \ell_2(x)\simeq A_++O(\lambda_2\ln x) ,
\end{equation}
and thus, as long as $\lambda_2\ln x$ remains small, we see that $\ell_2$ reaches a plateau value on the $x>0$ side. 
A posteriori, this tells us that the bath temperature and $L_+$ should be such that $L_+\ll \ee^{1/\lambda_2}$. 
Hence, the range of validity of our calculation is that for a given $L_+$, the bath temperature should be chosen so that $\ee^{\ee^{\beta\Delta V_b}}\gg L_+$. 
Of course, a similar reasoning holds on the left side with $L_-$ replacing $L_+$. 
In practice, this inequality is already achieved at moderate bath temperatures. In those conditions, and for all practical purposes, $\ell_2(x)$ saturates to a constant value in the vicinity of the walls, thus behaving in a way qualitatively similar to the $m>2$ case. 

Here, we use the connected harmonic-antiharmonic-harmonic potential to numerically calculate $\beta_\mathcal{M}$ as a function of $L_-$. The potential reads
\begin{equation}
V(x)=\begin{cases}
    \frac{\kappa}{2}(x+1)^2, & x<x_{m},\\
    V_{\text{M}}-\frac{\kappa}{2}x^2, & x_{m}\le x\le x_{p},\\
    -V_{\text{m}}+\frac{\kappa}{2}(x-\alpha)^2, & x>x_{p},
\end{cases}
\end{equation}
where $x_{m}$ and $x_{p}$ are the connecting points, $V_{\text{M}}$ and $V_{\text{m}}$ are the potential barrier height and the potential minimum, and $\alpha$ is the position of the minimum. The parameters $\alpha$, $x_{m}$, $x_{p}$ and $V_{\text{M}}$ are determined by the continuity of the potential and its first derivative at the connecting points. As shown in Fig.~\ref{fig:m=2}, we see that the scaling behavior of the Mpemba temperature with the position of the left wall also holds when $m=2$, including at moderate bath temperatures. 
\begin{figure}[t]
\begin{center}
\includegraphics[width=0.45\textwidth]{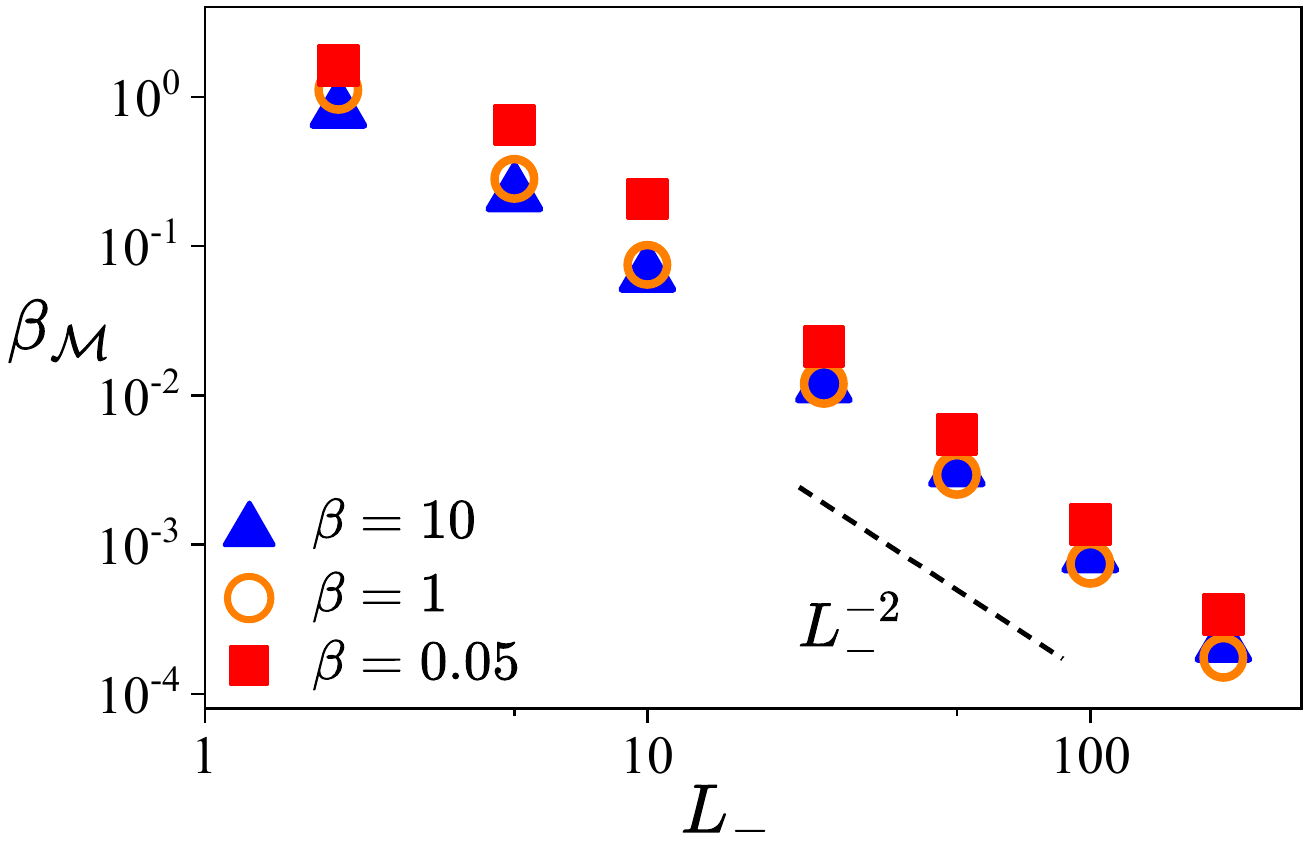}
\caption{The inverse temperature $\beta_\mathcal{M}$ is shown as a function of $L_-$ for various bath temperatures in a quadratic potential with a hard wall at $x=-L_-=-L_{\text{BL}}$. We have used $L_+=2L_-$ ($L_{\text{BR}}=2L_{\text{BL}}$), $\kappa=1$, $\alpha=3/2$, $x_{m}=-1/2$, $x_{p}=3/4$, $V_{\text{M}}=1/4$ and $V_{\text{m}}=5/8$. The dashed line is drawn to guide the eye and corresponds to the scaling $\beta_\mathcal{M}\sim L^{-2}_-$.
}
\label{fig:m=2}
\end{center}
\end{figure}

\section{Asymptotic behavior of the eigenmode $\ell_2(x)$ and the harmonic potential benchmark}
\label{app:ell2_behavior}

In this Appendix, we analyze the structural properties of the first non-trivial left eigenmode $\ell_2(x)$ governed by the adjoint Fokker-Planck operator. The exact eigenvalue equation is given by
\begin{equation}
T \frac{d^2 \ell_2(x)}{dx^2} - V'(x) \frac{d \ell_2(x)}{dx} = -\lambda_2 \ell_2(x),
\label{eq:app_eigen_eq}
\end{equation}
where $T$ represents the temperature and $\lambda_2 > 0$ is the relaxation rate of the slowest non-equilibrium mode. 

To evaluate whether a quadratic profile $\ell_2(x) \sim \text{const.} + A x^2$ holds universally, we must contrast the local behavior inside a weakly asymmetric barrier with the global behavior under a purely harmonic potential $V(x) = \kappa x^2/2$.

\subsection{The harmonic potential benchmark ($V(x) \sim x^2$)}

Let us consider a symmetric harmonic confinement defined by $V(x) = \kappa x^2/2$ with $\kappa > 0$. Substituting $V'(x) = \kappa x$ into Eq.~\eqref{eq:app_eigen_eq}, we obtain
\begin{equation}
T \ell_2''(x) - \kappa x \ell_2'(x) + \lambda_2 \ell_2(x) = 0.
\label{eq:app_harmonic_ode}
\end{equation}
By introducing the dimensionless variable $\xi = \frac{x}{\sqrt{2T/\kappa}}$, Eq.~\eqref{eq:app_harmonic_ode} can be transformed into Hermite's differential equation:
\begin{equation}
\frac{d^2 \ell_2}{d\xi^2} - 2\xi \frac{d \ell_2}{d\xi} + \frac{2\lambda_2}{\kappa} \ell_2 = 0.
\end{equation}
For the solution to remain physically bounded under the natural weight of the steady-state Boltzmann distribution $P_{\text{st}}(x) \propto e^{-V(x)/T}$, the eigenvalues must satisfy discrete quantization conditions. The first non-trivial relaxation rate corresponds to $n=1$, yielding:
\begin{equation}
\lambda_2 = \kappa.
\end{equation}
The corresponding left eigenmode is the first-order Hermite polynomial, which depends strictly linearly on the spatial coordinate:
\begin{equation}
\ell_2(x) = \sqrt{\frac{\kappa}{T}} x.
\label{eq:app_linear_profile}
\end{equation}
Equation~\eqref{eq:app_linear_profile} proves that for a globally harmonic potential $V(x) \sim x^2$, the eigenmode profile is linear across all space, rather than quadratic.

\subsection{Resolution of the quadratic local profile in asymmetric potentials}

The quadratic form $\ell_2(x) \sim \text{const.} + A x^2$ obtained in the main text is a specific, local consequence of expanding around a weakly asymmetric double-well potential, rather than an intrinsic property of arbitrary potentials $V(x)$. 

Consider a system where a reference potential $V_0(x)$ is purely symmetric, meaning the corresponding zeroth-order eigenmode $\ell_2^{(0)}(x)$ is strictly anti-symmetric. When we introduce a small, weakly asymmetric perturbation $\epsilon V_1(x)$, we can expand the eigenmode as $\ell_2(x) = \ell_2^{(0)}(x) + \epsilon \ell_2^{(1)}(x) + \mathcal{O}(\epsilon^2)$.

Inside the central barrier region ($x_- \le x \le x_+$) of a deep bistable landscape, the potential is locally flat or weakly curved, implying $V'(x) \simeq 0$. In this localized regime, Eq.~\eqref{eq:app_eigen_eq} simplifies to a diffusion-dominated equation:
\begin{equation}
T \frac{d^2 \ell_2(x)}{dx^2} \simeq -\lambda_2 \ell_2(x).
\end{equation}
Because the perturbation $\epsilon V_1(x)$ breaks the parity of the landscape, the first-order correction $\ell_2^{(1)}(x)$ must pick up a symmetric component to satisfy the modified boundary conditions across the wells. Integrating this local perturbation over the central region naturally yields the local quadratic profile:
\begin{equation}
\ell_2^{(1)}(x) = B_0 + A x^2, \quad \text{for } x \in [x_-, x_+],
\end{equation}
where $A \approx -\frac{\lambda_2}{2T} B_0$. 

In conclusion, the profile $\ell_2(x) \sim \text{const.} + A x^2$ is valid strictly as a local, perturbative description inside a weakly asymmetric barrier where $V'(x) \to 0$. It does not apply asymptotically to growing potentials, nor does it hold for a harmonic potential $V(x) \sim x^2$, where the true global profile is linear ($\ell_2(x) \propto x$).

%%%%%%%%%%%%%%
\section{Properties of parabolic cylinder functions}\label{app:Weber_Function}
%%%%%%%%%%%%%%%

The function $D_{\nu}(x)$ denotes the parabolic cylinder (Weber) function, which is defined as a solution of the differential equation
\begin{equation}
\frac{\dd^2 y}{\dd x^2} + \left(\nu + \frac{1}{2} - \frac{x^2}{4}\right) y = 0.
\end{equation}
A convenient explicit representation is given in terms of the confluent hypergeometric function ${}_1F_1$ as
\begin{align}\label{def:D_nu}
D_{\nu}(x)
&= 2^{\nu/2} \ee^{-x^2/4}
\Bigg[
\frac{\sqrt{\pi}}{\Gamma\!\left(\frac{1-\nu}{2}\right)}
\,{}_1F_1\!\left(-\frac{\nu}{2};\frac{1}{2};\frac{x^2}{2}\right)
\nonumber\\
&\qquad
- \frac{\sqrt{2\pi}\,x}{\Gamma\!\left(-\frac{\nu}{2}\right)}
\,{}_1F_1\!\left(\frac{1-\nu}{2};\frac{3}{2};\frac{x^2}{2}\right)
\Bigg],
\end{align}
where $\Gamma(x)=\int_0^\infty \dd t\,t^{x-1}\ee^{-t}$, and ${}_1F_1(a,b;z)$ is defined in Eq.~\eqref{def:M}.

%\begin{widetext}
%%%%%%%%%%%%%%%%%%%%%%%%%%%%%%%%%%%%%%%%%%%%%%%%%%%%%%%%%%%%%%%%%%%%%%%%%%%%%%%%%%%%%%%%%%%%%%%%%%%%
\section{Rigorous proof of the fixed sign bias of $\ell_2(x)$}\label{Fixed_sign_l2}
%%%%%%%%%%%%%%%%%%%%%%%%%%%%%%%%%%%%%%%%%%%%%%%%%%%%%%%%%%%%%%%%%%%%%%%%%%%%%%%%%%%%%%%%%%%%%%%%%%

This appendix gives a fully self-contained and detailed proof that, for an asymmetric double-well potential, the second left eigenfunction $\ell_2(x)$ of the Fokker-Planck operator has a \emph{fixed sign bias} between the left and right wells in the low-temperature regime.  
As mentioned in the main text, we assume that the region can be distinguished as the left well ($x<x_-)$, the barrier region ($x_-<x<x_+$), and the right well ($x>x_+$).
We also assume that the local minimum of $V(x)$ in the left well is higher than the global minimum of $V(x)$ at $x=x_+$ in the right well. 
The proof is constructive and uses only explicit properties of the potential and elementary Sturm-Liouville theory.

%%%%%%%%%%%%%%%%%%%%%%%%%%%%%%%%%%%%%%%%%%%%%%%%%%%%%%%%%%%%
\subsection{Low-temperature localization}
%%%%%%%%%%%%%%%%%%%%%%%%%%%%%%%%%%%%%%%%%%%%%%%%%%%%%%%%%%%%

As $\beta\to\infty$, the right well, which contains the global minimum, dominates the Boltzmann weight.  
Correspondingly, $\varphi_1(x)$ and $|\varphi_2(x)|$ are exponentially localized near $x=\alpha$ in the right well, with only exponentially small weight in the left well.

In contrast, the eigenfunction $\varphi_2(x)$ must be orthogonal to $\varphi_1(x)$:
\begin{equation}
    \int \dd x \varphi_1(x)\varphi_2(x)=0.
\label{eq:orth}
\end{equation}

Because $\varphi_1(x)$ is strictly positive and sharply localized in the right well, the integral in Eq.~\eqref{eq:orth} can vanish only if $\varphi_2(x)$ changes sign between the right-dominated region and the rest of the domain.

%%%%%%%%%%%%%%%%%%%%%%%%%%%%%%%%%%%%%%%%%%%%%%%%%%%%%%%%%%%%
\subsection{Location of the node}
%%%%%%%%%%%%%%%%%%%%%%%%%%%%%%%%%%%%%%%%%%%%%%%%%%%%%%%%%%%%

We now show that the unique node $x_0$ of $\varphi_2(x)$ lies in the central region $(x_-,x_+)$ for sufficiently large $\beta$.
Assume, for contradiction, that $x_0\le x_-$.  
Then $\varphi_2(x)$ has a fixed sign throughout the central region and the right well.  Since $\varphi_1(x)$ is exponentially small in those regions, the orthogonality condition \eqref{eq:orth} cannot be satisfied.
Similarly, assuming $x_0\ge x_+$ leads to the same contradiction.

Therefore, we conclude the inequality for the node satisfying
\begin{equation}
x_-<x_0<x_+.
\end{equation}

%%%%%%%%%%%%%%%%%%%%%%%%%%%%%%%%%%%%%%%%%%%%%%%%%%%%%%%%%%%%
\subsection{Fixed sign bias in the wells}
%%%%%%%%%%%%%%%%%%%%%%%%%%%%%%%%%%%%%%%%%%%%%%%%%%%%%%%%%%%%

Since the node lies strictly in the barrier (central) region, the sign of $\varphi_2(x)$ is constant
throughout each well:
\begin{equation}
    \varphi_2(x)
\begin{cases}
> 0 & x<x_0,\\
 <0 & x>x_0,
 \end{cases}
 \end{equation}
 up to an overall sign choice.

The left eigenfunction is
\begin{equation}
\ell_2(x)=\ee^{\beta V(x)/2}\varphi_2(x).
\end{equation}
Since the exponential factor is positive everywhere, $\ell_2$ has the same sign structure
as $\varphi_2$.

Hence, for an asymmetric double-well potential at sufficiently low temperature, $\ell_2(x)$ has a fixed sign throughout the left well and the opposite fixed sign throughout the right well.

%%%%%%%%%%%%%%%%%%%%%%%%%%%%%%%%%%%%%%%%%%%%%%%%%%%%%%%%%%%%
\subsection{Conclusion}
%%%%%%%%%%%%%%%%%%%%%%%%%%%%%%%%%%%%%%%%%%%%%%%%%%%%%%%%%%%%

For connected harmonic--antiharmonic--harmonic double-well potentials with a unique global
minimum, the second left eigenfunction $\ell_2(x)$ exhibits a rigorous and robust
left-right sign bias in the low-temperature regime.  
The result relies only on explicit properties of the potential, orthogonality to the ground state, and the Sturm-Liouville
nodal theorem, and does not invoke heuristic arguments.
%On the other hand, if we are interested in the motion in a weakly asymmetric potential, the monotonicity of $a_2$ against $\beta_i$ can be shown regardless of the form of the potential.

%\end{widetext}

%%%%%%%%%%%%%%%%

\section{Rigorous proof of the sign correlation of $V(x)-\langle V\rangle_{\beta_i}$}\label{Sign_DV}

%%%%%%%%%%%%%%%%%%%%%%%%%%%%%%

This appendix provides a mathematically explicit and self-contained proof of the sign correlation between the factor $V(x)-\langle V\rangle_{\beta_i}$ and the spatial location relative to the dominant probability weight.  
We clearly distinguish what can be proven \emph{generally} from what requires \emph{specific structural assumptions} on the potential.

%The discussion is intentionally detailed and computational, rather than heuristic.

%%%%%%%%%%%%%%%%%%%%%%%%%%%%%%%%%%%%%%%%%%%%%%%%%%%%%%%%%%%%
\subsection{Problem setting and statement}
%%%%%%%%%%%%%%%%%%%%%%%%%%%%%%%%%%%%%%%%%%%%%%%%%%%%%%%%%%%%

Let $V(x)$ be a smooth one-dimensional potential satisfying:
%\begin{enumerate}
(1) $V(x)\to+\infty$ as $|x|\to\infty$ (normalizability);
(2) $V(x)$ has exactly two local minima $x_-<x_+$ and one local maximum $x^*$;
and (3) $V(x_-)>V(x_+)$ (asymmetric double well).
%When we compare the setup with Eq.~\eqref{eq:piecewiseV}, $x_L=-1$, $x_R=\alpha$, and $x_M=0$.

By using the Boltzmann distribution $\pi(x,\beta)=\ee^{-\beta_i V(x)}/Z(\beta_i)$ and the thermal average
\begin{equation}
    \langle V\rangle_{\beta_i}=\int \dd x V(x)\pi(x,\beta_i) ,
\end{equation}
we can claim: For sufficiently large $\beta_i$ (low initial temperature), the sign of $V(x)-\langle V\rangle_{\beta_i}$ is (i) negative in a neighborhood of the global minimum $x_+$ and (ii) positive in a neighborhood of the higher-energy well and the barrier region.
We now prove this claim under increasing levels of generality.

%%%%%%%%%%%%%%%%%%%%%%%%%%%%%%%%%%%%%%%%%%%%%%%%%%%%%%%%%%%%
\subsection{General low-temperature expansion in the initial condition (Model-independent)}
%%%%%%%%%%%%%%%%%%%%%%%%%%%%%%%%%%%%%%%%%%%%%%%%%%%%%%%%%%%%

The unique global minimum is located at $x=x_+$.
Expanding $V(x)$ near $x_+$ yields
\begin{equation}
V(x)=V(x_+)+\frac{k_*}{2}(x-x_+)^2+O\left((x-x_+)^3\right),
\end{equation}
where $k_{*}>0$.
By Laplace's method, the partition function satisfies
\begin{equation}
Z(\beta_i)=\ee^{-\beta_i V(x_+)}\sqrt{\frac{2\pi}{\beta_i k_*}}
\left[1+O\left(\frac{1}{\beta_i}\right)\right].
\end{equation}
Similarly, the thermal average of the potential is
\begin{equation}
\langle V\rangle_{\beta_i}=V(x_+)+\frac{1}{2\beta_{i}}+O\left(\frac{1}{\beta_i^2}\right).
\label{eq:Vmean}
\end{equation}
Therefore, in a neighborhood $|x-x_+|\ll1$,
\begin{equation}
V(x)-\langle V\rangle_{\beta_i}
=\frac{k_*}{2}(x-x_+)^2-\frac{1}{2\beta_i}+O(|x-x_+|^3).
\end{equation}
For $|x-x_+|< (\beta_i k_*)^{-1/2}$, the right-hand side is strictly negative.
Thus, in the region where $\pi(x,\beta_i)$ is exponentially concentrated, $V(x)-\langle V\rangle_{\beta_i}<0$.

This conclusion holds for \emph{any} confining potential with a unique global minimum, if we use a low initial temperature.

%%%%%%%%%%%%%
\subsection{General condition for sign correlation}
%%%%%%%%%%%%%%%%%%%%%%

The previous argument was formulated in the limit of low initial temperature $\beta_i \to \infty$, which guarantees strong localization of the initial distribution around the global minimum.
However, the essential requirement is not the low temperature itself, but the localization property of the probability distribution $\pi(x,\beta_i)$.

More precisely, the argument remains valid whenever $\pi(x,\beta_i)$ is sufficiently concentrated in a neighborhood of the global minimum, i.e.,
\begin{equation}
\int_{|x-x_+|<d} \dd x\, \pi(x,\beta_i) \approx 1
\end{equation}
for some small $d$.

Under this condition, one still has
\begin{equation}
V(x) - \langle V \rangle_{\beta_i}
\begin{cases}
< 0 & \text{in the dominant region},\\
> 0 & \text{outside},
\end{cases}
\end{equation}
and the sign correlation used in the main text follows.

Thus, the assumption of low initial temperature can be replaced by a more general localization condition on the initial distribution. In practice, this is typically realized in the limit $\beta_i \gg 1$, but the argument does not strictly rely on this limit.

%%%%%%%%%%%%%%%%%%%%%%%%%%%%%%%%%%%%%%%%%%%%%%%%%%%%%%%%%%%%
\subsection{Contribution from the secondary well}
%%%%%%%%%%%%%%%%%%%%%%%%%%%%%%%%%%%%%%%%%%%%%%%%%%%%%%%%%%%%

Let $x_-$ be the local minimum of the higher-energy well, with
$\Delta V=V(x_-)-V(x_+)>0$.
Then
\begin{equation}
p_{\beta_i}(x_-)\sim \ee^{-\beta_i \Delta V}\ll p_{\beta_{i}}(x_+).
\end{equation}

In a neighborhood of $x_-$,
\begin{equation}
V(x)-\langle V\rangle_{\beta_i}
= \Delta V + O\left((x-x_-)^2\right) - O\left(\frac{1}{\beta_i}\right) > 0
\end{equation}
for sufficiently large $\beta_i$.

Thus, regions that are \emph{not} dominant in probability necessarily correspond to
positive values of $V(x)-\langle V\rangle_{\beta_i}$.

%%%%%%%%%%%%%%%%%%%%%%%%%%%%%%%%%%%%%%%%%%%%%%%%%%%%%%%%%%%%
\subsection{Summary of applicability}
%%%%%%%%%%%%%%%%%%%%%%%%%%%%%%%%%%%%%%%%%%%%%%%%%%%%%%%%%%%%

%\begin{itemize}
%\item 
The sign correlation holds \emph{generally} for any smooth confining potential
with a unique global minimum in the low initial temperature regime.
%\item 
For asymmetric double wells, it holds uniformly in regions that dominate the
Boltzmann weight.
%\item 
%For connected harmonic--antiharmonic--harmonic potentials, the statement can be
%proven by explicit calculation without approximation beyond large $\beta$.
%\end{itemize}
The argument is therefore not merely physical intuition but a consequence of precise asymptotic estimates.

%%%%%%%%%%%%%%%%
\section{Proof of the existence of the Mpemba temperature}\label{app:transition_temp}
%%%%%%%%%%%%%%%

In this Appendix, we provide a proof for the existence of the Mpemba temperature $\beta_{\mathcal{M}}$ in an asymmetric double-well potential with a wall on the shallow side, which is on the domain $[-L_-,\infty]$. Let $x_-$ and $x_+$ denote the local minima of the potential, satisfying $x_- < x^* < x_+$, where $x^*$ is the position of the barrier saddle point. 

\begin{theorem}
Consider any double-well potential $V(x)$ that satisfies: (i) $V(x)$ diverges fast enough such that $\int_{x^*}^\infty\dd{x}\, \ee^{-\beta V(x)}<\infty$ for any $\beta>0$, and (ii) $V(x_-)>V(x_+)$. Then, the following equation always has at least one solution $\beta_i>0$:
\begin{equation}
    \frac{\int_{-L_{-}}^{x^*}\dd x\,V(x)\ee^{-\beta_i V(x)}}{\int_{-L_{-}}^{x^*}\dd x\, \ee^{-\beta_i V(x)}} = \frac{\int_{x^*}^{\infty}\dd x\, V(x)\ee^{-\beta_i V(x)}}{\int_{x^*}^{\infty}\dd x\, \ee^{-\beta_i V(x)}}.
\end{equation}
\end{theorem}

\begin{proof}
For convenience, we define the following function of $\beta_i$:
\begin{equation}
    \delta U(\beta_i)\coloneqq \frac{\int_{-L_{-}}^{x^*}\dd{x}\,V(x)\ee^{-\beta_i V(x)}}{\int_{-L_{-}}^{x^*}\dd{x}\,\ee^{-\beta_i V(x)}} - \frac{\int_{x^*}^{\infty}\dd{x}\,V(x)\ee^{-\beta_i V(x)}}{\int_{x^*}^{\infty}\dd{x}\,\ee^{-\beta_i V(x)}}.
\end{equation}
Note that
\begin{align}
    \frac{\int_{-L_{-}}^{x^*}\dd{x}\,V(x)\ee^{-\beta_i V(x)}}{\int_{-L_{-}}^{x^*}\dd{x}\,\ee^{-\beta_i V(x)}}&\le \max_{-L_{-}\le x\le x^*}V(x),\\
    \lim_{\beta_i\to 0}\frac{\int_{x^*}^{\infty}\dd{x}\,V(x)\ee^{-\beta_i V(x)}}{\int_{x^*}^{\infty}\dd{x}\,\ee^{-\beta_i V(x)}}&=\infty,~\text{(see Prop.~\ref{prop:1})}\\
    \lim_{\beta_i\to\infty}\frac{\int_{-L_{-}}^{x^*}\dd{x}\,V(x)\ee^{-\beta_i V(x)}}{\int_{-L_{-}}^{x^*}\dd{x}\,\ee^{-\beta_i V(x)}}&=V(x_-),\\
    \lim_{\beta_i\to\infty}\frac{\int_{x^*}^{\infty}\dd{x}\,V(x)\ee^{-\beta_i V(x)}}{\int_{x^*}^{\infty}\dd{x}\,\ee^{-\beta_i V(x)}}&=V(x_+).
\end{align}
Evaluating $\delta U(\beta_i)$ at two limits $\beta_i\to 0$ and $\beta_i\to\infty$, we obtain
\begin{align}
    \lim_{\beta_i\to 0}\delta U(\beta_i)&=-\infty,\\
    \lim_{\beta_i\to \infty}\delta U(\beta_i)&=V(x_-)-V(x_+)>0.
\end{align}
From the intermediate value theorem, it follows that the equation $\delta U(\beta_i)=0$ has at least one nontrivial solution $\beta_i>0$.

\end{proof}

\begin{proposition}\label{prop:1}
For any potential $V(x)$ satisfying the conditions above, the following equality holds:
\begin{equation}
    \lim_{\beta_i\to 0}\frac{\int_{x^*}^{\infty}\dd{x}\,V(x)e^{-\beta_i V(x)}}{\int_{x^*}^{\infty}\dd{x}\,e^{-\beta_i V(x)}}=\infty.
\end{equation}

\end{proposition}

\begin{proof}
Without loss of generality, we can assume that $V(x)\ge 0~\forall x\ge x^*$ (if not, we can shift the potential as $V(x)\leftarrow V(x) - \min_{x\ge x^*}V(x)$ and the analysis remains unchanged).
For any $M>0$, consider $\beta_i=1/\max_{x\in[z,M+z]}|V(x)|$.
Then,
\begin{equation}
    \int_{z}^{\infty}\dd{x}\,\ee^{-\beta_i V(x)}>\int_{z}^{M+z}\dd{x}\,\ee^{-\beta_i V(x)}\ge M\ee^{-1}.
\end{equation}
Therefore,
\begin{equation}
    \lim_{\beta_i\to 0}\int_{z}^{\infty}\dd{x}\,\ee^{-\beta_i V(x)}=\infty,~~\forall z.
\end{equation}
For any $M>0$, let $x_M>x_*$ be the smallest number such that $V(x)\ge M$ for any $x\ge x_M$.
Then, we can evaluate as follows:
\begin{equation}
    \begin{split}
    &\frac{\int_{x^*}^{\infty}\dd{x}\,V(x)\ee^{-\beta_i V(x)}}{\int_{x^*}^{\infty}\dd{x}\,\ee^{-\beta_i V(x)}}\\
    &\ge\frac{M\int_{x_M}^{\infty}\dd{x}\,\ee^{-\beta_i V(x)}}{\int_{x^*}^{x_M}\dd{x}\,\ee^{-\beta_i V(x)}+\int_{x_M}^{\infty}\dd{x}\,\ee^{-\beta_i V(x)}}\\ &\ge \frac{M\int_{x_M}^{\infty}\dd{x}\,\ee^{-\beta_i V(x)}}{x_M-x_*+\int_{x_M}^{\infty}\dd{x}\,\ee^{-\beta_i V(x)}}.
    \end{split}
\end{equation}
Since $\lim_{\beta_i\to 0}\int_{x_M}^{\infty}\dd{x}\,\ee^{-\beta_i V(x)}=\infty$, one can find $\beta_i'>0$ such that
\begin{equation}
    \int_{x_M}^{\infty}\dd{x}\,\ee^{-\beta_i V(x)}\ge x_M-x_*,~~\forall \beta_i\le\beta_i'.
\end{equation}
Consequently, the following inequality holds for any $\beta_i\le\beta_i'$:
\begin{equation}
    \frac{\int_{x^*}^{\infty}\dd{x}\,V(x)e^{-\beta_i V(x)}}{\int_{x^*}^{\infty}\dd{x}\,e^{-\beta_i V(x)}}\ge\frac{M}{2}.
\end{equation}
Therefore, 
\begin{equation}
    \lim_{\beta_i\to 0}\frac{\int_{x^*}^{\infty}\dd{x}\,V(x)\ee^{-\beta_i V(x)}}{\int_{x^*}^{\infty}\dd{x}\,\ee^{-\beta_i V(x)}}=\infty.
\end{equation}

\end{proof}

\end{appendix}

%%%%%%%%%%%%%%%%%%%%%%%%%%%%%%%%%%%%%%%%%%%%%%%%%%%%%%%%%%%%%%%%%%%%%%%%%%%%%%%%%
%\bibliography{mpemba-gemini}

\begin{thebibliography}{99}
%%%%%%%%%%%%%%%%%%%%%%%%%%%%%%%%%%%%%%%%%%%%%%%%

\bibitem{mpemba1969cool}
E. B. Mpemba and D. G. Osborne, Cool?
	\href{https://doi.org/10.1088/0031-9120/4/3/312}{Phys. Educ. {\bf 4}, 172 (1969)}.
\bibitem{burridge2016questioning} 
H. C. Burridge and P. F. Linden, Questioning the Mpemba effect: hot water does not cool more quickly
than cold, 
\href{https://doi.org/10.1038/srep37665}{Sci. Rep. \textbf{6}, 37665 (2016)}.

\bibitem{katz2017reply}
J. I. Katz, Reply to Burridge $\&$ Linden: Hot water may freeze sooner than cold, 
\href{https://arxiv.org/abs/1701.03219}{arXiv preprint arXiv:1701.03219 (2017)}.


%%%%%%%%%%%%%%%%%%%%%%%%
\bibitem{Teza25}
G. Teza, J. Bechhoefer, A. Lasanta, O. Raz, and M. Vucelja,
Speedups in nonequilibrium thermal relaxation: Mpemba and related effects,
\href{https://doi.org/10.1016/j.physrep.2025.10.009}{Phys. Rep. \textbf{1164}, 1 (2026)}.

\bibitem{Ares25_review}
F. Ares,  P. Calabrese, and S. Murciano, The quantum Mpemba effects,
\href{https://doi.org/10.1038/s42254-025-00838-0}{Nat. Rev. Phys. \textbf{7}, 451 (2025)}.

\bibitem{Bechhoefer2021}
J. Bechhoefer, A. Kumar, and R. Ch\'{e}trite, A fresh understanding of the Mpemba effect, \href{https://doi.org/10.1038/s42254-021-00349-8}{Nat. Rev. Phys. \textbf{3}, 534 (2021)}.

%%%%%%%%%%%%%%%%%%%%%%%%
\bibitem{Lu17}
	Z. Lu and O. Raz, 
	Nonequilibrium Thermodynamics of the Markovian Mpemba Effect and Its Inverse, 
	\href{https://doi.org/10.1073/pnas.1701264114}
	{Proc. Natl. Acad. Sci. U.S.A. {\bf 114}, 5083 (2017)}.

\bibitem{Klich19}
	I. Klich, O. Raz, O. Hirschberg, and M. Vucelja, 
    Mpemba Index and Anomalous Relaxation,
	\href{https://doi.org/10.1103/PhysRevX.9.021060}
	{Phys. Rev. X {\bf 9}, 021060 (2019)}.

\bibitem{Busiello21}
D. M. Busiello, D. Gupta, and A. Maritan, 
Inducing and optimizing Markovian Mpemba effect with stochastic reset,
\href{https://iopscience.iop.org/article/10.1088/1367-2630/ac2922/meta}{New J. Phys. {\bf 23}, 103012 (2021)}.



%%%%%%%%%%%%%%%%%
\bibitem{kumar2020}
A. Kumar and J. Bechhoefer, 
Exponentially faster cooling in a colloidal system,
\href{https://www.nature.com/articles/s41586-020-2560-x}{Nature {\bf 584},  64 (2020)}.

\bibitem{kumar2021}
A. Kumar, Anomalous relaxation in colloidal systems,
\href{https://summit.sfu.ca/item/34663}{PhD thesis, Simon Fraser University, 2021.}

\bibitem{Kumar22}
A. Kumar, R. Ch\'{e}trite and J. Bechhoefer,
Anomalous heating in a colloidal system,
\href{https://doi.org/10.1073/pnas.2118484119}{Proc. Natl. Acad. Sci. U.S.A.  {\bf 119}, e2118484119 (2022).}	


\bibitem{Malhotra24}
I. Malhotra and H. L\"{o}wen, Double Mpemba effect in the cooling of trapped colloids,
\href{https://doi.org/10.1209/0295-5075/ac8573}{J. Chem. Phys. \textbf{161}, 164903 (2024). }

\bibitem{Tian2025}
Y. Tian, Y. Zheng, L.-H. Liu, L. Wang, G.-C. Guo, and F.-W. Sun, Experimental study of Mpemba effect in an energy Langevin system, \href{https://journals.aps.org/prresearch/abstract/10.1103/5p4l-1515}{Phys. Rev. Research \textbf{7}, L042020 (2025).}

%%%%%%%%%%%%%%%



\bibitem{Santos17}
	A. Lasanta, F. Vega Reyes, A. Prados, and A. Santos, 
    When the Hotter Cools More Quickly: Mpemba Effect in Granular Fluids,
	\href{https://doi.org/10.1103/PhysRevLett.119.148001}
	{Phys. Rev. Lett. {\bf 119}, 148001 (2017)}.

		

\bibitem{Torrente19}
	A. Torrente, M. A. L\'{o}pez-Casta\~{n}o, A. Lasanta, F. Vega Reyes, A. Prados, and A. Santos, 
    Large Mpemba-like effect in a gas of inelastic rough hard spheres,
	\href{https://doi.org/10.1103/PhysRevE.99.060901}
	{Phys. Rev. E {\bf 99}, 060901(R) (2019)}.
\bibitem{Biswas20}
	A. Biswas, V. V. Prasad, O. Raz, and R. Rajesh,
    Mpemba effect in driven granular gas,
	\href{https://doi.org/10.1103/PhysRevE.102.012906}
	{Phys. Rev. E {\bf 102}, 012906 (2020)}.
\bibitem{Biswas21}
	A. Biswas, V. V. Prasad, and R. Rajesh,
    Mpemba effect in an anisotropically driven granular Maxwell gases,
	\href{https://doi.org/10.1209/0295-5075/ac2d54}
	{EPL {\bf 136}, 46001 (2021)}.	
\bibitem{Mompo20}
	E. Momp\'{o}, M. A. L\'{o}pez Casta\~{n}o, A. Lasanta, F. Vega Reyes, and A. Torrente,
    Memory effects in a gas of viscoelastic particles,
	\href{https://doi.org/10.1063/5.0050804}
	{Phys. Fluids {\bf 33}, 062005 (2021)}.


\bibitem{Megias22b}
A. Meg\'{i}as and A. Santos, 
Mpemba-like effect protocol for granular gases of inelastic and rough hard disks,
\href{https://doi.org/10.3389/fphy.2022.971671}{Front. Phys. {\bf 10}, 971671 (2022)}.


\bibitem{Patron23}
A. Patr\'{o}n, B. S\'{a}nchez-Rey, C. A. Plata and A. Prados, 
Non-equilibrium memory effects: Granular fluids and beyond,
\href{https://doi.org/10.1209/0295-5075/acf7e5}{EPL {\bf 143}, 61002 (2023)}.
\bibitem{Keller18}
	T. Keller, V. Torggler, S. B. J\"{a}ger, S. Sch\"{u}tz, H. Ritsch, and G. Morigi, 
Quenches across the self-organization transition in multimode cavities, 
	\href{https://doi.org/10.1088/1367-2630/aaa161}
	{New J. Phys. {\bf 20}, 025004 (2018)}.	

\bibitem{Santos20}
	A. Santos and A. Prados,
    Mpemba effect in molecular gases under nonlinear drag,
	\href{https://doi.org/10.1063/5.0016243}
	{Phys. Fluids {\bf 32}, 072010 (2020)}.

\bibitem{Patron21}
 A. Patr\'{o}n, B. S\'{a}nchez-Rey and A. Prados, 
 Strong nonexponential relaxation and memory effects in a fluid with nonlinear drag,
\href{https://doi.org.org/10.1103/PhysRevE.104.064127}{Phys. Rev. E {\bf 104}, 064127 (2021).}	
	
\bibitem{Takada21a}
S. Takada, H. Hayakawa, and A. Santos, 
Mpemba effect in inertial suspensions,
\href{https://doi.org.org/10.1103/PhysRevE.103.032901}{Phys. Rev. E {\bf 103}, 032901 (2021).}	

%\bibitem{Takada21b}
%S. Takada, \href{https://doi.org/10.1051/epjconf/202124904001}{EPJ Web. Conf. {\bf 249}, 04001 (2021). }

\bibitem{Santos24}
A. Santos, 
Mpemba meets Newton: Exploring the Mpemba and Kovacs effects in the time-delayed cooling law,
\href{https://doi.org/10.1103/PhysRevE.109.044149}{Phys. Rev. E \textbf{109}, 044149 (2024)}.


\bibitem{Lin22}
J. Lin, K. Li, J. He, J. Ren and J. Wang,
Power statistics of Otto heat engines with the Mpemba effect,
\href{https://doi.org/10.1103/PhysRevE.105.014104}{Phys. Rev. E {\bf 105}, 014104 (2022).}

\bibitem{Pal24}
A. Biswas and A. Pal,
Mpemba effect on nonequilibrium active Markov chains,
\href{https://doi.org/10.1103/PhysRevE.111.054136}{Phys. Rev. E \textbf{111}, 054136 (2025)}.

\bibitem{Biswas23a}
A. Biswas, R. Rajesh and A. Pal,  
Mpemba effect in a Langevin system: Population statistics, metastability, and other exact results, 
\href{https://doi.org/10.1063/5.0155855}{J. Chem. Phys. {\bf 159}, 044120 (2023)}.

\bibitem{Biswas_thesis}
A. Biswas, 
Mpemba effect in granular and Langevin systems, 
\href{https://www.imsc.res.in/xmlui/bitstream/handle/123456789/619/HBNI-TH-228.pdf?sequence=1&isAllowed=y}{The thesis of The Institute of Mathematical Sciences, Chennai}.


\bibitem{chetrite2021metastable}
R. Chétrite, A. Kumar, and J. Bechhoefer, 
The Metastable Mpemba Effect Corresponds to a Non-monotonic Temperature Dependence of Extractable Work, 
\href{https://doi.org/10.3389/fphy.2021.654271}{Frontiers in Physics \textbf{9}, 654271 (2021)}.

\bibitem{walker2021anomalous}
M. R. Walker and M. Vucelja, 
Anomalous thermal relaxation of Langevin particles in a piecewise-constant potential, 
\href{https://iopscience.iop.org/article/10.1088/1742-5468/ac2edc}{J. Stat. Mech.: Theory and Experiment 2021, 113105 (2021)}.

\bibitem{Deguenther22}
J. Deg\"{u}nther and U. Seifert,
Anomalous relaxation from a non-equilibrium steady state: An isothermal analog of the Mpemba effect,
\href{https://iopscience.iop.org/article/10.1209/0295-5075/ac8573}{EPL \textbf{139}, 41002 (2022).}





\bibitem{Yang20}
Z.-Y. Yang and J.-X. Hou, 
Non-Markovian Mpemba effect in mean-field systems,
\href{https://doi.org/10.1103/PhysRevE.101.052106}{Phys. Rev. E {\bf 101}, 052106 (2020).}


\bibitem{Strachan25}
D. J. Strachan, A.  Purkayastha, and S. R. Clark,
Non-Markovian Quantum Mpemba Effect,
\href{https://doi.org/10.1103/PhysRevLett.134.220403}{Phys. Rev. Lett. \textbf{134}, 220403 (2025)}.

\bibitem{Yang22}
Z.-Y. Yang and J.-X. Hou, 
Mpemba effect of a mean-field system: The phase transition time,
\href{https://doi.org/10.1103/PhysRevE.105.014119}{Phys. Rev. E {\bf 105}, 014119 (2022).}




\bibitem{Gonzalez2021}
R. G. Gonz\'{a}lez, N. Khalil and V. Garz\'{o}, 
Mpemba-like effect in driven binary mixtures,
\href{https://doi.org/10.1063/5.0050530}{Phys. Fluids {\bf 33}, 053301 (2021)}.

\bibitem{Baity-Jesi19}
	M. Baity-Jesi, E. Calore, A. Cruz, L. A. Fernandez, J. M. Gil-Narvi\'{o}n, A. Gordillo-Guerrero, D. I\~{n}iguez, A. Lasanta, A. Maiorano, E. Marinari, V. Martin-Mayor, J. Moreno-Gordo, A. Mu\~{n}oz Sudupe, D. Navarro, G. Parisi, S. Perez-Gaviro, F. Ricci-Tersenghi, J. J. Ruiz-Lorenzo, S. F. Schifano, B. Seoane, A. Taranc\'{o}n, R. Tripiccione, and D. Yllanes,
    The Mpemba effect in spin glasses is a persistent memory effect,
	\href{https://doi.org/10.1073/pnas.1819803116}{Proc. Natl. Acad. Sci. U.S.A. {\bf 116}, 15350 (2019)}.



\bibitem{Greaney11}
	P. A. Greaney, G. Lani, G. Cicero, and J. C. Grossman,
    Mpemba-Like Behavior in Carbon Nanotube Resonators,
\href{https://doi.org/10.1007/s11661-011-0843-4}
	{Metall. Mater. Trans. A {\bf 42}, 3907 (2011)}.
	

\bibitem{Gonzalez21}
	I. Gonz\'{a}lez-Adalid Pemrt\'{i}n, E. Momp\'{o}, A. Lasanta, V. Mart\'{i}n-Mayor and J. Salas, 
    Slow growth of magnetic domains helps fast evolution routes for out-of-equilibrium dynamics,
\href{https://doi.org/10.1103/PhysRevE.104.044114}{Phys. Rev. E {\bf 104}, 044114 (2021)}.

%\bibitem{Chetrite21}
%R. Che\'{e}trite, A. Kumar and J. Bechhoefer, \href{https://doi.org/10.3389/fphy.2021.654271}{Front. Phys. {\bf 9}, 654271 (2021).}

\bibitem{Holtzman22}
R. Holtzman and O. Raz, 
Landau theory for the Mpemba effect through phase transitions,
\href{https://doi.org/10.1038/s42005-022-01063-2}{Communications  Physics {\bf 5}, 280 (2022)}.

\bibitem{Megias22a}
A. Meg\'{i}as, A. Santos, and A. Prados, 
Thermal versus entropic Mpemba effect in molecular gases with nonlinear drag,
\href{https://doi.org/10.1103/PhysRevE.105.054140}{Phys. Rev. E {\bf 105}, 054140 (2022)}.

\bibitem{Biswas23}
A. Biswas, V. V. Prasad, and R. Rajesh, 
Mpemba effect in driven granular gases: Role of distance measures,
\href{https://doi.org/10.1103/PhysRevE.108.024902}{Phys. Rev. E \textbf{108}, 024902 (2023).}

\bibitem{Pemartin23}
I. G. Pemartin, E. Momp\'{o}, A. Lasanta, V. Martin-Mayor and  J. Salas, 
Shortcuts of Freely Relaxing Systems Using Equilibrium Physical Observables,
\href{https://doi.org/10.1103/PhysRevLett.132.117102}{Phys. Rev. Lett. \textbf{132}, 117102 (2024)}.


\bibitem{Antonov26}
A. P. Antonov  and H. L\"{o}wen, 
Temperature overshooting in the Mpemba effect of frictional active matter,
\href{https://link.aps.org/doi/10.1103/j4wb-1d9d}{Phys. Rev. E \textbf{113}, 025407 (2026). }

%%%%%%%%%%
\bibitem{Ohga2024}
N. Ohga, H. Hayakawa, and S. Ito, 
Microscopic theory of Mpemba effects and a no-Mpemba theorem for monotone many-body systems,
\href{https://doi.org/10.48550/arXiv.2410.06623}{arXiv:2410.06623}.

\bibitem{Vu2025}
T. Van Vu and H. Hayakawa, 
Thermomajorization Mpemba Effect,
\href{https://doi.org/10.1103/PhysRevLett.134.107101}{Phys. Rev. Lett. \textbf{134}, 107101 (2025)}.
%%%%%%%%%%%%%


\bibitem{Nava19}
A. Nava and M. Fabrizio,
Lindblad dissipative dynamics in the presence of phase coexistence, 
\href{https://doi.org/10.1103/PhysRevB.100.125102}{Phys. Rev. B. {\bf 100}, 125102 (2019).}

\bibitem{Carollo21}
F. Carollo, A. Lasanta, and I. Lesanovsky,
Exponentially Accelerated Approach to Stationarity in Markovian Open Quantum Systems through the Mpemba Effect,
\href{https://doi.org/10.1103/PhysRevLett.127.060401}{Phys. Rev. Lett. {\bf 127}, 060401 (2021).}

\bibitem{Manikandan21}
S. K. Manikandan, 
Equidistant quenches in few-level quantum systems, 
\href{https://doi.org/10.1103/PhysRevResearch.3.043108}{Phys. Rev. Res. {\bf 3}, 043108 (2021)}.

\bibitem{Ivander23}
F. Ivander, N. Anto-Sztrikacs, and D. Segal, 
Hyperacceleration of quantum thermalization dynamics by bypassing long-lived coherences: An analytical treatment,
\href{https://doi.org/10.1103/PhysRevE.108.014130}{Phys. Rev. E \textbf{108}, 014130 (2023)}.

\bibitem{Ares23}
F. Ares, S. Murciano, and P. Calabrese, 
Entanglement asymmetry as a probe of symmetry breaking,
\href{https://doi.org/10.1038/s41467-023-37747-8}{Nat. Commun. {\bf 14}, 2036 (2023)}.

\bibitem{Chatterjee_2023}
A. K. Chatterjee, S. Takada, and H. Hayakawa, 
Quantum Mpemba Effect in a Quantum Dot with Reservoirs,
\href{https://doi.org/10.1103/PhysRevLett.131.080402}{Phys. Rev. Lett. {\bf 131}, 080402 (2023)}.

\bibitem{Chatterjee24}
A. K. Chatterjee, S. Takada, and H. Hayakawa, 
Multiple quantum Mpemba effect: Exceptional points and oscillations,
\href{https://doi.org/10.1103/PhysRevA.110.022213}{Phys. Rev. A \textbf{110}, 022213 (2024)}.

\bibitem{Joshi2024} 
L. K. Joshi, J. Franke, A. Rath, F. Ares, S. Murciano, F. Kranzl, R. Blatt, P. Zoller, B. Vermersch, P. Calabrese, C. F. Roos, and M. K. Joshi,
Observing the quantum Mpemba effect in quantum simulations,
\href{https://doi.org/10.1103/PhysRevLett.133.010402}{Phys. Rev. Lett. \textbf{133}, 010402 (2024)}.

\bibitem{Shapira2024}
S. A. Shapira, Y. Shapira, J. Markov, G. Teza, N. Akerman, O. Raz, and R. Ozeri,
Inverse Mpemba Effect Demonstrated on a Single Trapped Ion Qubit,
\href{https://doi.org/10.1103/PhysRevLett.133.010403}{Phys. Rev. Lett. \textbf{133}, 010403 (2024)}.

\bibitem{Rylands2024}
C. Rylands, K. Klobas, F. Ares, P. Calabrese, S. Murciano, and B. Bertini, 
Microscopic Origin of the Quantum Mpemba Effect in Integrable Systems,
\href{https://doi.org/10.1103/PhysRevLett.133.010401}{Phys. Rev. Lett. \textbf{133}, 010401 (2024).}


\bibitem{Moroder24}
M. Moroder, O. Culhane, K. Zawadzki, and J. Goold, 
Thermodynamics of the quantum Mpemba effect,
\href{https://doi.org/10.1103/PhysRevLett.133.140404}{Phys. Rev. Lett. \textbf{133}, 140404 (2024)}.

\bibitem{Chalas24}
K. Chalas, F. Ares, C. Rylands, and P. Calabrese,
Multiple crossing during dynamical symmetry restoration and implications for the quantum Mpemba effect,
\href{https://doi.org/10.1088/1742-5468/ad769c}{J. Stat. Mech. (2024) 103101}.

\bibitem{Ares24}
F. Ares, V. Vitale, and Sara Murciano, 
The quantum Mpemba effect in free-fermionic mixed states,
\href{https://doi.org.org/10.1103/PhysRevB.111.104312}{Phys. Rev. B \textbf{111}, 104312 (2025)}.

\bibitem{Yamashika2024}
S. Yamashika, F. Ares, and P. Calabrese, 
Entanglement asymmetry and quantum Mpemba effect in two-dimensional free-fermion systems,
\href{https://doi.org/10.1103/PhysRevB.110.085126}{Phys. Rev. B \textbf{110}, 085126 (2024).}


\bibitem{Liu2024}
S. Liu, H.-K. Zhang, S. Yin, and S.-X. Zhang, 
Symmetry Restoration and Quantum Mpemba Effect in Symmetric Random Circuits,
\href{https://doi.org/10.1103/PhysRevLett.133.140405}{Phys. Rev. Lett. \textbf{133}, 140405 (2024).}

\bibitem{Wang2024}
X. Wang and J. Wang, Mpemba effects in nonequilibrium open quantum systems,
\href{https://doi.org/10.1103/PhysRevResearch.6.033330}{Phys. Rev. Research, \textbf{6}, 033330 (2024).}

\bibitem{Nava2024}
A. Nava and R. Egger, Mpemba Effects in Open Nonequilibrium Quantum Systems,
\href{https://doi.org/10.1103/PhysRevLett.133.136302}{Phys. Rev. Lett. \textbf{133}, 136302 (2024).}

\bibitem{Longhi2024}
S. Longhi, Photonic Mpemba effect,
\href{https://doi.org/10.1364/OL.532503}{Opt. Lett. \textbf{49}, 5188 (2024).}



%%%%%%%%%%%%%%%%%%
\bibitem{Zhang2025}
J. Zhang, G. Xia, C.-W. Wu, T. Chen, Q. Zhang, Y. Xie, W.-B. Su, W. Wu, C.-W. Qiu, P.-X. Chen, W. Li, H. Jing, and Y.-L. Zhou, Observation of quantum strong Mpemba effect, 
\href{https://doi.org/10.1038/s41467-024-54303-0}{Nat. Commun. \textbf{16}, 301 (2025).}

\bibitem{Turkeshi2025}
X. Turkeshi, P. Calabrese, and A. De Luca, Quantum Mpemba Effect in Random Circuits, 
\href{https://link.aps.org/doi/10.1103/5d6p-8d1b}{Phys. Rev. Lett. \textbf{135}, 040403 (2025)}. 


\bibitem{Bao2025}
R. Bao and Z. Hou, Accelerating Quantum Relaxation via Temporary Reset: A Mpemba-Inspired Approach, 
\href{https://link.aps.org/doi/10.1103/g94p-7421}{Phys. Rev. Lett. \textbf{135}, 150403 (2025).} 

\bibitem{Yu2025}
H. Yu, S. Liu, and S.-X. Zhang, Quantum mpemba effects from symmetry perspectives, 
\href{https://doi.org/10.1007/s43673-025-00157-7}{AAPPS Bulletin \textbf{35}, 17 (2025).}

\bibitem{Beato2026} 
N. Beato and G. Teza, Relaxation Control of Open Quantum Systems, 
\href{https://link.aps.org/doi/10.1103/4frd-ck2z}{Phys. Rev. Lett. \textbf{136}, 070401 (2026).} 

\bibitem{Yamashika2026}
S. Yamashika and F. Ares, 
Quantum Mpemba Effect in Long-Range Spin Systems, 
\href{https://link.aps.org/doi/10.1103/52y5-8kl2}{Phys. Rev. Lett. \textbf{136}, 090402 (2026).} 






\bibitem{Yue26}
Y. Liu, T. Van Vu, R. Chétrite, F. van Wijland, and H. Hayakawa, 
The Mpemba effect likes to hit a wall,
\href{https://doi.org/10.48550/arXiv.2604.01543}{arXiv:2604.01543}.

\bibitem{Sagawa}
    T. Sagawa, 
    Entropy, Divergence, and Majorization in Classical and Quantum Thermodynamics 
    (Springer-Brief in Mathematical Physics, Springer, Singapore, 2022).
    See also 
    \href{https://doi.org/10.48550/arXiv.2007.09974}
    {arXiv:2007.09974}.


\bibitem{Hayakawa2026}
H. Hayakawa and S. Takada, 
Mpemba effect in a two-dimensional bistable potential,
\href{https://doi.org/10.48550/arXiv.2603.24148}{arXiv:2603.24148}.

%%%%%%%%%%%%%%%%%%%%%%%%%%%%

\bibitem{risken1996fokker} 
H. Risken, Fokker-planck equation, 2nd ed. (Springer, Berlin, 1996).


\bibitem{brey1990}
J. J. Brey and J. Casado, 
Generalized Langevin equations with time-dependent temperature, 
\href{https://doi.org/10.1007/BF01027298}{J. Stat. Phys. \textbf{61}, 713 (1990)}.

\bibitem{TynMyint-U-1978}
TynMyint-U, Ordinary Differential Equations (North-Holland, New York, 1978).

\bibitem{Sane26}
S. S. Sane and J. Bechhoefer,
Mpemba effect without a wall
\href{https://doi.org/10.48550/arXiv.2609.32871}{arXiv:2609.32781}.

\bibitem{fortuin1971correlation}
C. M. Fortuin, P. W. Kasteleyn, and J. Ginibre, Correlation inequalities on some partially ordered sets, 
\href{https://doi.org/10.1007/BF01651330}{Commun. Math. Phys., 22, 89 (1971).} 


%%%%%%%%%%%%%%%%%%%%%%%%
\bibitem{matkowsky1977exit}
B. J. Matkowsky and Z. Schuss, 
The exit problem for randomly perturbed dynamical systems,
\href{https://doi.org/10.1137/0133024}{SIAM J. Appl. Math. \textbf{33}, 365 (1977).}

\bibitem{bovier2004metastability}
A. Bovier, M. Eckhoff, V. Gayrard, M. Klein, 
Metastability in Reversible Diffusion Processes I: Sharp Asymptotics for Capacities and Exit Times,
\href{https://doi.org/10.4171/JEMS/14}{J. Eur. Math. Soc. \textbf{6}, 399 (2004).}

%%%%%%%%%%%%%%%%%%%%%%%

%\bibitem{Abhay2024}
% A. Srivastav,  V. Pandey, B. Mohan, and A. K. Pati,
% Family of exact and inexact quantum speed limits for completely positive and trace-preserving dynamics,
% \href{https://doi.org/10.1103/npkr-c4vf}{Phys. Rev. A \textbf{112}, 052204 (2025)}.

\bibitem{Abramowitz}
    M. Abramowitz and I. A. Stegun, 
    Handbook of Mathematical Functions: With Formulas, Graphs, and Mathematical Tables 
    (Dover Publications, New York, 1964).
%%%%%%%%%%%






\end{thebibliography}

\end{document}